\documentclass{article}

\usepackage{adjustbox}

\usepackage[final]{neurips_2025}

\usepackage[utf8]{inputenc} 
\usepackage[T1]{fontenc}    
\usepackage{hyperref}       
\usepackage{url}            
\usepackage{booktabs}       
\usepackage{amsfonts}       
\usepackage{nicefrac}       
\usepackage{microtype}      
\usepackage{xcolor}         
\usepackage[table]{xcolor}
\usepackage{amsmath}
\usepackage{amssymb}
\usepackage{array}
\usepackage{amsthm}
\usepackage{graphicx}
\usepackage{multirow}
\usepackage{bbding}
\usepackage{tabularx}

\newtheorem{definition}{Definition}
\newtheorem{assumption}{Assumption}
\newtheorem{lemma}{Lemma}
\newtheorem{theorem}{Theorem}
\usepackage[skip=4pt]{caption}
\PassOptionsToPackage{numbers, compress}{natbib}

\title{\textsc{Sneakdoor}: Stealthy Backdoor Attacks against Distribution Matching-based Dataset Condensation}

\author{%
  He Yang\textsuperscript{1,2}\thanks{Equal contribution.},
  Dongyi Lv\textsuperscript{1}\footnotemark[1],
  Song Ma\textsuperscript{1,2},
  Wei Xi\textsuperscript{1,2}\thanks{Corresponding author.},
  Jizhong Zhao\textsuperscript{1,2} \\
  \textsuperscript{1} School of Computer Science and Technology, Xi'an Jiaotong University, Xi'an, China \\
  \textsuperscript{2} National Key Laboratory of Human-Machine Hybrid Augmented Intelligence, \\
  Xi'an Jiaotong University, Xi'an, China \\
  \texttt{yanghe73@xjtu.edu.cn}, \texttt{\{lvdongyi, song.ma\}@stu.xjtu.edu.cn}, \\
  \texttt{\{xiwei, zjz\}@xjtu.edu.cn}
}

\begin{document}
\maketitle

\begin{abstract}
  Dataset condensation aims to synthesize compact yet informative datasets that retain the training efficacy of full-scale data, offering substantial gains in efficiency. Recent studies reveal that the condensation process can be vulnerable to backdoor attacks, where malicious triggers are injected into the condensation dataset, manipulating model behavior during inference. While prior approaches have made progress in balancing attack success rate and clean test accuracy, they often fall short in preserving stealthiness, especially in concealing the visual artifacts of condensed data or the perturbations introduced during inference. To address this challenge, we introduce \textsc{Sneakdoor}, which enhances stealthiness without compromising attack effectiveness. \textsc{Sneakdoor} exploits the inherent vulnerability of class decision boundaries and incorporates a generative module that constructs input-aware triggers aligned with local feature geometry, thereby minimizing detectability. This joint design enables the attack to remain imperceptible to both human inspection and statistical detection. Extensive experiments across multiple datasets demonstrate that \textsc{Sneakdoor} achieves a compelling balance among attack success rate, clean test accuracy, and stealthiness, substantially improving the invisibility of both the synthetic data and triggered samples while maintaining high attack efficacy. The code is available at \url{https://github.com/XJTU-AI-Lab/SneakDoor}.
\end{abstract}

\section{Introduction}
Dataset Condensation (DC)~\cite{liu2023slimmable,hemultisize,zhao2023dm,yang2023efficient,yin2023squeeze,zhang2024m3d} has recently emerged as a powerful paradigm for synthesizing compact training datasets that retain the learning efficacy of their full-sized counterparts, offering substantial benefits in terms of computation, memory, and deployment efficiency. However, DC introduces inherent vulnerabilities to backdoor attacks~\cite{li2022backdoor,guo2022overview,li2022backdoors,wu2022backdoorbench}, where malicious triggers can be injected into the distilled samples during the condensation process. Once compromised, the distilled dataset can disseminate malicious behaviors across downstream models, undermining model integrity and posing serious security threats.

A growing body of work demonstrates that malicious triggers, once implanted into the distilled set, can persist across downstream training and inference, leading to consistent and targeted misclassification~\cite{liu2023backdoor,chung2024rethinking,zheng2023rdm}. One of the earliest approaches is the Naive Attack~\cite{liu2023backdoor}, which directly adds a fixed visual pattern (typically a static patch) to instances from clean training samples before condensation. While conceptually simple, this method suffers from limited attack success rates, as the uniform trigger tends to degrade through the condensation process. To enhance attack effectiveness, Doorping~\cite{liu2023backdoor} introduces a bilevel optimization framework that iteratively updates both the distilled data and the backdoor trigger during training. Doorping better preserves the trigger semantics and achieves stronger attack success rate. However, it incurs significant computational cost due to its bilevel nature and lacks a theoretical foundation. A more recent work~\cite{chung2024rethinking} adopts a kernel-theoretic lens to reinterpret backdoor vulnerability in condensation. They propose two variants, simple-trigger and relax-trigger. The former attack focuses exclusively on minimizing the generalization gap, aiming to ensure that the backdoor learned during condensation reliably transfers to test-time behavior. The relax-trigger introduces a joint optimization objective that simultaneously reduces projection loss (mismatch between synthetic and clean distributions), conflict loss (interference between clean and poisoned instances), and the generalization gap. Notably, relax-trigger maintains high attack success rate while avoiding the computational overhead of bilevel optimization.

\begin{figure}
    \centering
    \includegraphics[width=0.8\linewidth]{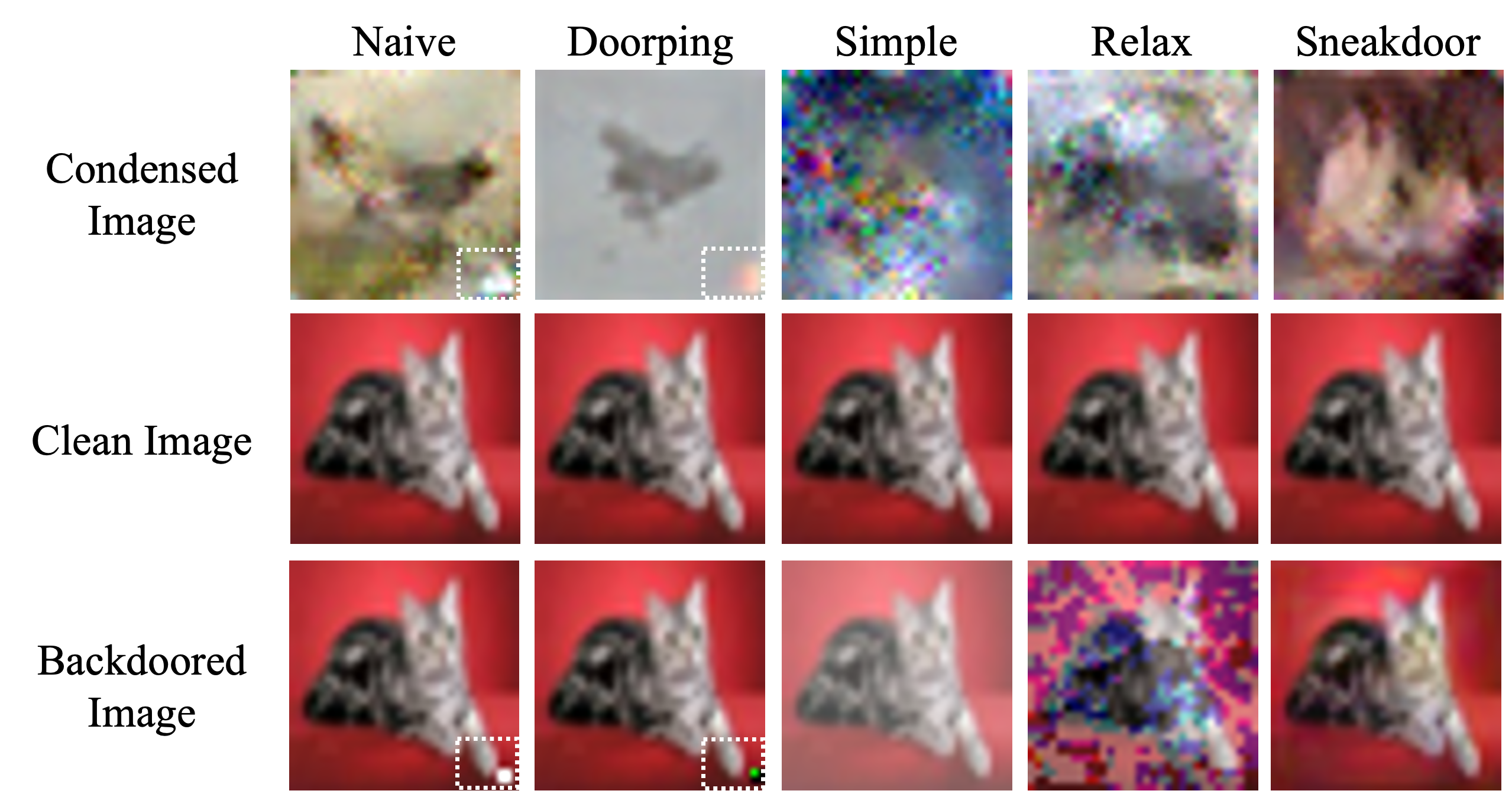}
    \caption{Stealthiness Illustration}
    \label{fig:example}
    \vspace{-6pt}
\end{figure}

However, existing approaches fall short of achieving a well-calibrated trade-off among attack success rate (ASR), clean test accuracy (CTA), and stealthiness (STE). While some methods attain high ASR or maintain acceptable CTA, they frequently neglect STE, a critical dimension that reflects the visual and statistical imperceptibility of both the distilled data and the triggered inputs (See Figure~\ref{fig:example}). This oversight is particularly damaging, without sufficient stealthiness, even highly effective attacks become vulnerable to detection, significantly limiting their practical viability. This persistent imbalance motivates our proposed method, \textsc{Sneakdoor}, which leverages input-aware trigger generation and decision boundary sensitivity, achieving a more favorable balance among ASR, CTA, and STE.

Specifically, \textsc{Sneakdoor} consists of two stages, (1) Trigger Generation and (2) Backdoor Injection. In the first stage, a generative network is trained to produce input-aware triggers tailored to individual samples. By aligning each trigger with the local semantic content of its host image, the perturbations remain visually coherent and difficult to isolate. In the second stage, the backdoor injection is formulated as an optimization problem. The generated triggers are embedded into a subset of clean samples to form a poisoned subset. These triggered samples are then incorporated into the training set prior to condensation, allowing the distilled dataset to encode backdoor behavior alongside clean task representations. As a result, downstream models trained on the synthesized data exhibit the intended malicious behavior without sacrificing generalization to clean inputs.

Our contributions are summarized below:

\begin{itemize}
    \item 
    We present the first investigation of backdoor attacks against distribution matching-based dataset condensation, with a focus on jointly optimizing ASR, CTA, and STE.
    \item 
    We provide a theoretical analysis of stealthiness concerning \textsc{Sneakdoor}, offering formal guarantees and insights into the conditions under which backdoor signals remain undetectable throughout the condensation and training process.
    \item 
    Extensive experiments across six datasets demonstrate that \textsc{Sneakdoor} consistently outperforms existing methods in achieving a superior balance across ASR, CTA, and STE.
\end{itemize}

\section{Related Work}
\textbf{Distribution Matching-based Dataset Condensation:} Dataset condensation (DC) aims to synthesize a compact set of synthetic samples that can replace large-scale datasets while preserving comparable model performance. Among various condensation paradigms, distribution matching (DM)-based methods have emerged as a leading approach due to their scalability, generality, and empirical effectiveness. Unlike earlier techniques based on gradient matching or training trajectory alignment, DM-based methods directly align statistical or feature-level distributions between real and synthetic data. A seminal example is DM~\cite{zhao2023dm}, which matches the second-order moments (covariance) of feature embeddings extracted by random encoders. A core formulation in distribution matching-based dataset condensation leverages the maximum mean discrepancy (MMD) to quantify the distance between the feature distributions of real and synthetic samples in a high-dimensional embedding space. The objective is to minimize this discrepancy over the synthetic set \(\mathcal{S}\), ensuring statistical alignment with the original dataset \(\mathcal{T}\). Specifically, the optimization problem is defined as: $\min_{\mathcal{S}}\mathbb{E}_{\boldsymbol{\theta} \sim P_{\boldsymbol{\theta}}} \| \frac{1}{|\mathcal{T}|} \sum_{i=1}^{|\mathcal{T}|} \psi_{\boldsymbol{\vartheta}}(\mathcal{A}(\boldsymbol{x}_i, \omega)) - \frac{1}{|\mathcal{S}|} \sum_{j=1}^{|\mathcal{S}|} \psi_{\boldsymbol{\vartheta}}(\mathcal{A}(\boldsymbol{s}_j, \omega)) \|^2$, where \(\psi_{\boldsymbol{\vartheta}}\) is a randomly initialized and fixed embedding function, and \(\mathcal{A}(\cdot, \omega)\) denotes a differentiable Siamese augmentation operator applied to both real and synthetic samples, parameterized by \(\omega\). This formulation encourages the synthetic set to preserve the statistical structure of the real dataset under randomized transformations, thereby promoting generalization across model initializations drawn from \(P_{\boldsymbol{\theta}}\).

Subsequent extensions, such as IDM and DAM, enhance class-conditional alignment through kernel-based moment matching, adaptive feature regularization, and encoder updates, yielding improved performance. IDM introduces practical enhancements to the original distribution matching framework, incorporating progressive feature extractor updates, stronger data augmentations, and dynamic class balancing to improve generalization. In parallel, DataDAM leverages attention map alignment to better preserve spatial semantics, guiding synthetic samples to activate similar regions as real data while maintaining computational efficiency. These methods advance the state of dataset condensation by demonstrating that richer supervision and adaptive training dynamics are critical for generating high-fidelity synthetic datasets.

\textbf{Backdoor Attacks against Dataset Condensation:} Backdoor attacks aim to manipulate model behavior at inference time by injecting carefully crafted triggers into a subset of training data. When effective, the model performs normally on clean inputs but consistently misclassifies inputs containing the trigger. While extensively studied in standard supervised learning, backdoor attacks in the context of dataset condensation have only recently received attention. A pioneering study by Liu et al.~\cite{liu2023backdoor} introduces backdoors by poisoning real data before dataset condensation. Their Naive Attack appends a fixed trigger to target-class samples before condensation, but suffers from trigger degradation and reduced attack efficacy due to the synthesis process. To address this, Doorping employs a bilevel optimization scheme that jointly refines the trigger and the synthetic data. Although more effective, it incurs substantial computational overhead. More recently, Chung et al.~\cite{chung2024rethinking} provide a kernel-theoretic perspective on backdoor persistence in condensation. They propose simple-trigger, which minimizes the generalization gap of the backdoor effect, and relax-trigger, which further reduces projection and conflict losses for improved robustness. 

Importantly, existing approaches focus predominantly on maximizing ASR or preserving CTA, often overlooking STE, which is a critical factor for realistic attacks. In contrast, we propose \textsc{Sneakdoor}, a novel framework that explicitly addresses the ASR–CTA–STE trade-off through input-aware trigger generation and stealth-aware integration into distribution matching-based condensation.

\section{Methodology}
\subsection{Threat Model}
\textbf{Attack Scenario.}
We consider a \textit{collaborative setting} where one entity possesses a high-quality dataset and shares a compact version with another party via dataset condensation, due to privacy or bandwidth constraints. The condensed dataset is typically regarded as a trustworthy proxy for training. However, this trust can be exploited. A malicious provider, with full access to the original data and sole control over the condensation process, can embed backdoor triggers into the synthetic data. These triggers, while preserving high utility for clean tasks, can cause targeted misclassification in downstream models. 

Moreover, our threat model does \emph{not} assume that the attacker knows the downstream (victim) model architecture. This upstream threat underscores a critical vulnerability: even limited data sharing can serve as a potent attack vector when the condensation process is adversarially controlled. 

\textbf{Attacker’s Goal.} The attacker’s objective in backdooring condensed datasets is inherently multi-faceted, requiring a delicate balance among three goals: stealthiness (STE), attack success rate (ASR), and clean test accuracy (CTA). Due to space constraints, detailed definitions of these metrics are provided in Appendix A.

\subsection{Stealthy Backdoor Attack against Dataset Condensation}
(1) \textit{Trigger Generation}

Trigger generation starts by identifying the source–target class pair $(i,j)$ with the highest inter-class misclassification rate:
\begin{equation}
    \mathcal{O}_{i\to j}=\frac{1}{N}\sum_{k=1}^N \mathbb{I}\big(g_{\theta_c}(f_{\theta_f}(x_k))=j\big), \quad x_k\in\mathcal{T}_i,
    \label{eq:mc_bv}
\end{equation}
where $\mathcal{T}_i$ represents the subset of the original dataset $\mathcal{T}$ with ground-truth label $i$, $f_{\theta_f}$ and $g_{\theta_c}$ denote the feature extractor and classifier, respectively, $\mathbb{I}(\cdot)$ is the indicator function that equals $1$ if the classifier assigns the sample $x_k$ to class $j$, and $0$ otherwise. In practice, we estimate $\mathcal{O}_{i\to j}$ by sampling $N$ examples from class $i$, mapping them to the latent space with $f_{\theta_f}$, and computing the fraction that $g_{\theta_c}$ assigns to class $j$.

We evaluate $\mathcal{O}_{i\to j}$ for all ordered class pairs and select the pair with the maximal value. The chosen pair indicates the most error-prone direction for label confusion; a trigger is then designed to exploit this specific weakness. By targeting the pair with highest misclassification rate, the attack achieves consistent source→target misclassification while limiting collateral impact on overall model accuracy.




The computation of $\mathcal{O}_{i \rightarrow j}$ depends on the model parameters $\theta = \{\theta_f, \theta_c\}$, which correspond to the feature extractor $f_{\theta_f}$ and the classifier $f_{\theta_c}$, respectively. To obtain these parameters, we first construct a condensed dataset $\mathcal{S} = \{(x'_i, y'_i)\}_{i=1}^N$ from the original dataset $\mathcal{T} = \{(x_i, y_i)\}_{i=1}^M$, where $N \ll M$. The synthetic dataset $\mathcal{S}$ is generated by minimizing a distribution-matching objective over randomly initialized models, ensuring that training on $\mathcal{S}$ approximates the behavior of models trained on the full dataset $\mathcal{T}$:

\begin{equation}
\begin{aligned}
\mathcal{S}^* = \underset{\mathcal{S}}{\arg\min}\;
\mathbb{E}_{x \sim p_\mathcal{T},\, x' \sim p_{\mathcal{S}},\, \theta \sim p_\theta}
D\!\left(P_{\mathcal{T}}(x; \theta),\, P_{\mathcal{S}}(x'; \theta)\right)
+ \lambda\,\mathcal{R}(\mathcal{S}),
\end{aligned}
\end{equation}

where $P_{\mathcal{T}}(x;\theta)$ and $P_{\mathcal{S}}(x';\theta)$ denote the feature distributions induced by the original and condensed datasets, respectively. The distance measure $D(\cdot,\cdot)$, such as Maximum Mean Discrepancy (MMD), quantifies the discrepancy between these distributions. $\mathcal{R}(S)$ is a regularization term, and $\lambda$ balances the trade-off between distribution alignment and regularization.

After generating the condensed dataset $\mathcal{S}$, we train a surrogate model parameterized by $\theta = \{\theta_f, \theta_c\}$ using only $\mathcal{S}$. This surrogate serves as an efficient approximation of the downstream model’s decision behavior. Once trained, it is evaluated on the original dataset $\mathcal{T}$, and a normalized confusion matrix is computed to analyze inter-class prediction tendencies.
\begin{equation}
\begin{aligned}
C&=\frac{C_{ij}}{\sum_{j=0}^{o_c-1}C_{ij}}\\
C_{ij}&=\sum_{(x,y)\in\mathcal{T}}\mathbb{I}[y=i]\mathbb{I}[g_{\theta_c}(f_{\theta_f})(x)=j]
\end{aligned}
\end{equation}
where $o_c$ is the total number of classes in the original dataset $\mathcal{T}$. $C_{ij}$ represents the empirical probability that a sample from class $i$ is misclassified as class $j$. The maximum inter-class misclassification rate $\mathcal{O}_{y_s\xrightarrow{}y_\tau}$ is then calculated as follows:
\begin{equation}
    \mathcal{O}_{y_s\xrightarrow{}y_\tau}=\underset{i,j}{\arg\max C_{ij}},\quad i\neq j
\end{equation}
This measure identifies the class pair $(i,j)$ with the highest misclassification probability, revealing the most vulnerable decision boundary in the model. 

We then proceed to the trigger generation phase, where the objective is to create a trigger that, when added to an input sample, causes the model to misclassify the input from the source class $y_s$ to the target class $y_\tau$. Speicifically, we utilize a generator model $G_\phi$, which generates perturbations, or triggers, which are added to the original input data. The perturbation is designed to be imperceptible, ensuring the trigger remains stealthy while causing misclassification. The trigger generation process can be represented as follows:
\begin{equation}
    \label{eq:gen}
    \begin{aligned}
    &\widetilde{x}=x+\alpha G_\phi(x),\quad\forall x\in\mathcal{T}_{y_s}\\
    &s.t.\quad\|G_\phi(x)\|_\infty<\varepsilon,\quad\forall x
    \end{aligned}
\end{equation}

where $G_\phi(x)$ represents the generated adversarial noise, while $\varepsilon$ is a constraint that controls the maximum permissible perturbation, ensuring that the perturbation remains subtle and undetectable. The perturbed input is denoted as $\widetilde{x}$. The subset $\mathcal{T}_{y_s}$ refers to the portion of the original dataset for which the label is $y_s$. $\alpha$ is a small constant, further controlling the size of the perturbation.

In practice, the maximum permissible perturbation constraint in Eq.(\ref{eq:gen}) is enforced by applying a clamping operation to the generator output $G_\phi(x)$ before adding it to the original input. Specifically, the adversarial noise is clamped such that its $\ell_\infty$-norm lies within the range $[-\varepsilon, \varepsilon]$, ensuring the perturbation remains imperceptible. This clamped noise is then added to the clean image, followed by another clamping step to maintain the pixel values within the valid image range. The loss in Eq.(\ref{eq:G_loss}) is computed on these clamped, perturbed images, allowing the generator to be implicitly optimized under the perturbation constraint without the need for an explicit penalty term in the objective.

The generator model $G_\phi$ is trained alongside $\theta=\{\theta_f,\theta_c\}$, with the objective of minimizing the classification loss associated with the target class $y_\tau$. Specifically, the generator is updated based on the following objective function:
\begin{equation}
\label{eq:G_loss}
    \phi=\phi-\eta_\phi\sum_{x\in\mathcal{T}_{y_s}}\mathcal{L}\left(g_{\theta_c}(f_{\theta_f}(x+G_\phi(x))),y_\tau\right)
\end{equation}

where $\mathcal{L}$ is the loss function, which measures the error in predicting the target class $y_\tau$ after applying the trigger to the input $x$, and $\eta_\phi$ is the learning rate for the generator. 

By iteratively updating the generator, the generator $G_\phi$ is refined to produce more effective backdoor triggers. The process continues until the trigger causes consistent misclassifications of the source class $y_s$ as the target class $y_\tau$, while keeping the perturbation within the imperceptibility threshold $\varepsilon$. This approach enables the adversary to design highly effective backdoor triggers, leveraging the generator to produce stealthy perturbations that successfully compromise the performance of the downstream model.

(2) \textit{Backdoor Injection}

Once the generator $G_\phi$ has been trained to generate perturbations that cause misclassifications of the source class $y_s$ to the target class $y_\tau$, we proceed with the backdoor injection process. This step involves adding the learned perturbations to the source class samples in the original dataset $\mathcal{T}$. Specifically, we add the perturbations generated by $G_\phi$ to each sample $x \in \mathcal{T}_{y_s}$:
\begin{equation}
    \widetilde{x}=x+\alpha G_\phi(x)\quad\forall x\in\mathcal{T}_{y_s}
\end{equation}

where $\widetilde{x}$ represents the perturbed sample, and $G_\phi(x)$ is the perturbation generated by the adversarial generator. These perturbed samples are then relabeled to the target class $y_\tau$. 

This process ensures that adversarial perturbations are applied to the samples from the source class, resulting in a set of triggered samples, $\mathcal{T}_{\text{triggered}} = {(\widetilde{x}, y_\tau)}_{i=1}^{N_{\text{triggered}}}$, where the perturbed inputs are labeled as the target class $y_\tau$. In the subsequent step, the triggered samples are incorporated with the clean samples from the target class $y_\tau$. The primary objective of this combination is to introduce a fraction of the triggered samples into the target class, thereby facilitating the model to misclassify source class samples as the target class when subjected to the adversarial trigger. This process ensures that the model’s decision boundary is subtly manipulated to favor misclassification under specific conditions. Let $N_{\text{triggered}}$ be the total number of triggered samples generated in the previous step, each labeled with the target class $y_\tau$. The number of clean samples in the target class $y_\tau$ in the original dataset $\mathcal{T}{y_\tau}$ is denoted by $N_{\mathcal{T}{y_\tau}}$. Based on the poison ratio $\rho$, we will add $\rho \cdot N_{\mathcal{T}_{y_\tau}}$ triggered samples into $\mathcal{T}_{y_\tau}$. Specifically, we first randomly select $\rho \cdot N_{\mathcal{T}_{y_\tau}}$ samples from $\mathcal{T}_{\text{triggered}}$ and add them into $\mathcal{T}_{y_\tau}$. The resulting poisoned dataset $\mathcal{T}_{\text{mixed}}$ consists of both the clean target class samples and the triggered samples:
\begin{equation}
\mathcal{T}_\text{mixed}=\mathcal{T}_{y_\tau}\cup\{(\widetilde{x},y_\tau)\}_{i=1}^{\rho\cdot N_{\mathcal{T}_{y_\tau}}}
\end{equation}


The next step is to recondense the target class $\mathcal{T}_{y_\tau}$. The objective of recondensation is to generate a new subset $\mathcal{S}_{y_\tau}$ within the synthetic dataset, which preserves the key characteristics of the target class while amplifying the influence of the triggered samples. This process seeks to strike a balance between maintaining the intrinsic features of the target class and maximizing the impact of the adversarial samples. Specifically, the objective is to generate a synthetic dataset $\mathcal{S}{y_\tau}$ that closely approximates the target class distribution in the poisoned data $\mathcal{T}{y_\tau}$. The optimization objective is defined as:
\begin{equation}
\begin{aligned}
\mathcal{S}_{y_\tau}^*=\underset{\mathcal{S}_{y_\tau}}{\arg\min}\mathbb{E}_{x\sim p_{\mathcal{T}_\text{mixed}},x'\sim p_{\mathcal{S}_{y_\tau}},\theta\sim p_\theta}D\left(P_{\mathcal{T}_\text{mixed}}(x;\theta),P_{\mathcal{S}_{y_\tau}}(x';\theta)\right)+\lambda\mathcal{R}(\mathcal{S}_{y_\tau})
\end{aligned}
\end{equation}

where $P_{\mathcal{T}_\text{mixed}}(x;\theta)$ is the probability distribution of the target class incorporating triggered samples. $P_{\mathcal{S}_{y_\tau}}(x';\theta)$ is the probability distribution of the recondensed target class. 

\section{Stealthiness Analysis}
A critical challenge in designing effective backdoor attacks on dataset condensation is achieving stealthiness, ensuring that poisoned samples and the resulting synthetic data are indistinguishable from their clean counterparts. Our goal is to formalize stealthiness through a geometric and distributional lens, grounded in the feature space induced by deep neural architectures.

To this end, our analysis is guided by the following question: How does input-aware backdoor injection perturb the structure of data manifolds in feature space, and can this deviation be rigorously bounded to guarantee stealth? Since distribution matching-based condensation aligns global feature statistics (\textit{e.g.}, moments of embedded data), it is essential to understand whether triggers introduce detectable geometric or statistical anomalies in the condensed representation. We conduct our analysis in a Reproducing Kernel Hilbert Space (RKHS), where class-specific data, both clean and triggered, are assumed to lie on smooth, locally compact manifolds. By modeling the trigger as a bounded, input-aware perturbation and invoking assumptions on manifold regularity and inter-class proximity, we show that triggered samples remain tightly coupled to the clean data manifold under mild conditions. This theoretical framework enables us to quantify the effect of poisoning both at the feature level (Theorem~\ref{theorem:upper bound on feature-manifold}) and at the level of the condensed dataset (Theorem~\ref{theorem:MMD}). These results provide principled justification for \textsc{Sneakdoor}’s empirical stealth: the perturbations introduced by the trigger remain latent-space-aligned and distributionally consistent, limiting their detectability after condensation.

\textit{Formal statements of assumptions, intermediate lemmas, and proofs supporting our theoretical analysis are deferred to Appendix B for clarity and completeness.}

\begin{definition}[Kernel]
    $k:\mathcal{X}\times\mathcal{X}\mapsto\mathbb{R}$ on a non-empty set $\mathcal{X}$ is a kernel if it satisfies the following two conditions: (1) \textit{symmetry}: $k(x,x')=k(x',x),\quad\forall x,x'\in\mathcal{X}$. (2) Positive Semi-Definiteness: for any finite subset $\{x_1,x_2,\cdots,x_n\}\subset\mathcal{X}$, the Gram matrix $\mathbf{K}=[k(x_i,x_j)]_{i,j=1}^n$ is positive semi-definite.
\end{definition}

\begin{definition}[Reproducing Kernel Hilbert Space, RKHS]
    Given a kernel $k:\mathcal{X}\times\mathcal{X}\mapsto\mathbb{R}$, the Reproducing Kernel Hilbert Space $\mathcal{H}_k$ is a Hilbert space of functions $f:\mathcal{X}\mapsto\mathbb{R}$ satisfying: (1) For every $x\in\mathcal{X}$, the function $k(x,\cdot)\in\mathcal{H}_k$. (2) $\forall x\in\mathcal{X}$ and $f\in\mathcal{H}_k$, $f(x)=\langle f,k(x,\cdot)\rangle_{\mathcal{H}_k}$.
\end{definition}

\begin{theorem}[Upper Bound on Feature-Manifold Deviation under Poisoning]
Let \( \mathcal{T}_{y_\tau} \) denote the clean target-class dataset and \( \mathcal{T}_{\mathrm{triggered}} \) the triggered (poisoned) dataset, with corresponding feature-space distributions \( P_{\mathcal{M}_{\mathrm{clean}}} \) and \( P_{\mathcal{M}_{\mathrm{triggered}}} \), respectively. Define the mixed distribution as: $P_{\mathcal{M}_{\mathrm{mixed}}} = (1-\rho) P_{\mathcal{M}_{\mathrm{clean}}} + \rho P_{\mathcal{M}_{\mathrm{triggered}}}$, where \( \rho \in [0,1] \) denotes the poisoning ratio. Under Assumptions 1 (Lipschitz Continuity), 2 (Local Compactness of Feature Manifold), and 3 (Inter-Class Hausdorff Distance), the expected deviation of samples from the mixed distribution to the target feature manifold satisfies:
\begin{equation}
    \mathbb{E}_{z \sim P_{\mathcal{M}_{\mathrm{mixed}}}} \left[ \inf_{z_\tau \in \mathcal{M}_{\mathrm{clean}}} \|z - z_\tau\|_{\mathcal{H}} \right] \leq \rho (\gamma \varepsilon + \delta),
\end{equation}

where \( \mathcal{H} \) is the RKHS associated with the feature encoder.
\label{theorem:upper bound on feature-manifold}
\end{theorem}


\begin{theorem}[Upper Bound on the Discrepancy Between Poisoned and Clean Condensation Datasets]
Let \( \mathcal{T}_{y_\tau} \) denote the clean target-class dataset and \( \mathcal{T}_{\mathrm{mixed}} = \mathcal{T}_{y_\tau} \cup \mathcal{T}_{\mathrm{triggered}} \), where \( \mathcal{T}_{\mathrm{triggered}} \) consists of source-class samples \( x \in \mathcal{T}_{y_s} \) perturbed by a trigger generator \( G_\phi \) and relabeled as the target class. Let \( \mathcal{S}_{\mathrm{clean}} \) and \( \mathcal{S}_{\mathrm{poison}} \) denote the condensation datasets distilled from \( \mathcal{T}_{y_\tau} \) and \( \mathcal{T}_{\mathrm{mixed}} \), respectively, by minimizing: $
    \mathcal{S}^* = \arg\min_{\mathcal{S}} \mathrm{MMD}(\mathcal{T}, \mathcal{S}) + \lambda \mathcal{R}(\mathcal{S})$, where \( \mathcal{T} \in \{ \mathcal{T}_{y_\tau}, \mathcal{T}_{\mathrm{mixed}} \} \), \( \lambda > 0 \), and \( \mathcal{R} \) is a $\mu_R$ strongly convex regularizer. Under Assumptions 1 (Lipschitz Continuity), 2 (Local Compactness of Feature Manifold), and 3 (Inter-Class Hausdorff Distance), the MMD between \( \mathcal{S}_{\mathrm{clean}} \) and \( \mathcal{S}_{\mathrm{poison}} \) satisfies:
\[
\mathrm{MMD}(\mathcal{S}_{\mathrm{clean}}, \mathcal{S}_{\mathrm{poison}}) \leq \frac{L_f^2 \rho (\gamma \varepsilon + \delta)}{\lambda \mu_R}
\]
where \( \gamma = L_f \alpha \), \( \delta = \sup_{z_s \in \mathcal{M}_{\mathrm{source}}} \inf_{z_\tau \in \mathcal{M}_{\mathrm{clean}}} \|z_s - z_\tau\|_{\mathcal{H}} \), \( \rho \) is the poisoning rate, and \( \varepsilon \) bounds the input perturbation.
\label{theorem:MMD}
\end{theorem} 

\section{Experiments}

\paragraph{Datasets and Networks.}
We evaluate \textsc{Sneakdoor} across five standard datasets: FMNIST~\cite{xiao2017fashion}, CIFAR-10~\cite{krizhevsky2009learning}, SVHN~\cite{netzer2011reading}, Tiny-ImageNet~\cite{le2015tiny}, STL-10~\cite{coates2011analysis}, and ImageNette~\cite{howard2020fastai}. These datasets span a diverse range of visual complexity, semantic granularity, and image resolution, enabling a comprehensive evaluation of attack generality. Each dataset is processed according to the standard dataset condensation protocol, with 50 images per class used for condensation. Specifically, we adopt two common synthetic data backbones: {ConvNet} and {AlexNetBN}~\cite{NIPS2012_c399862d}, which represent lightweight and moderately expressive condensation encoders. For downstream training and evaluation, we consider four architectures: ConvNet, AlexNetBN, VGG11~\cite{simonyan2014very}, and ResNet18~\cite{he2016deep}. Moreover, we evaluate \textsc{Sneakdoor} in comparison with four state-of-the-art attacks: {NAIVE}~\cite{liu2023backdoor}, {DOORPING}~\cite{liu2023backdoor}, {SIMPLE}~\cite{chung2024rethinking}, and {RELAX}~\cite{chung2024rethinking}. 


\paragraph{Evaluation Metrics.}
We evaluate attack performance across three key dimensions: {ASR}, {CTA}, and {STE}. Following prior work~\cite{nips24waveattack}, {STE} is quantified using three complementary metrics: (1) PSNR (Peak Signal-to-Noise Ratio), measuring pixel-level similarity between triggered and clean samples, where higher values indicate lower perceptual distortion. (2) SSIM (Structural Similarity Index), which measures structural similarity, with values closer to 1 indicating stronger visual alignment; and (3) IS (Inception Score) quantifies the KL divergence between the predicted label distribution of a sample and the marginal distribution over all samples. Lower IS values suggest reduced recognizability, indicating higher stealth and improved resistance to detection. For convenience, we define an inverted score $\text{IS}^{\dagger} = (10^{-3} - \text{IS})\text{e}^{-4}$, where larger values correspond to improved stealth.

\paragraph{Overall Attack Effectiveness.}

We first evaluate the overall effectiveness of each backdoor attack in balancing three key objectives: {ASR}, {CTA}, and {STE}. To illustrate this trade-off, we visualize the normalized performance of each method using radar plots (Figure~\ref{fig:radar_balance_STL10}, Figure~\ref{fig:radar_balance_Tiny}) that jointly capture all three dimensions. \textsc{Sneakdoor} consistently achieves a superior balance across the three criteria. In contrast, while Doorping and Relax achieve high ASR, they suffer from significant degradation in either CTA or STE. Conversely, Naive and Simple maintain better CTA but fail to deliver competitive ASR or STE. These results validate our central hypothesis: \emph{input-aware trigger design combined with distribution-aligned injection enables the attack that is both effective and stealthy}. 


\begin{figure}[h]
    \centering
    \includegraphics[width=0.9\textwidth]{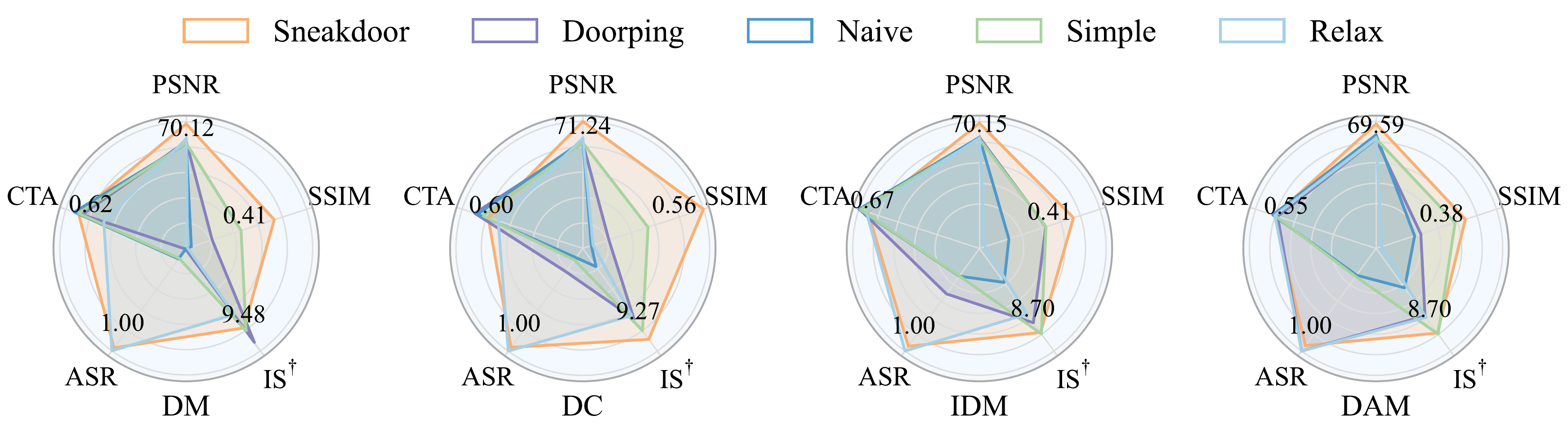}
    \caption{Attack Performance on STL10. Larger area indicates better balance.}
    \label{fig:radar_balance_STL10}
\end{figure}
\begingroup
\begin{figure}[h]
    \centering
    \includegraphics[width=0.9\textwidth]{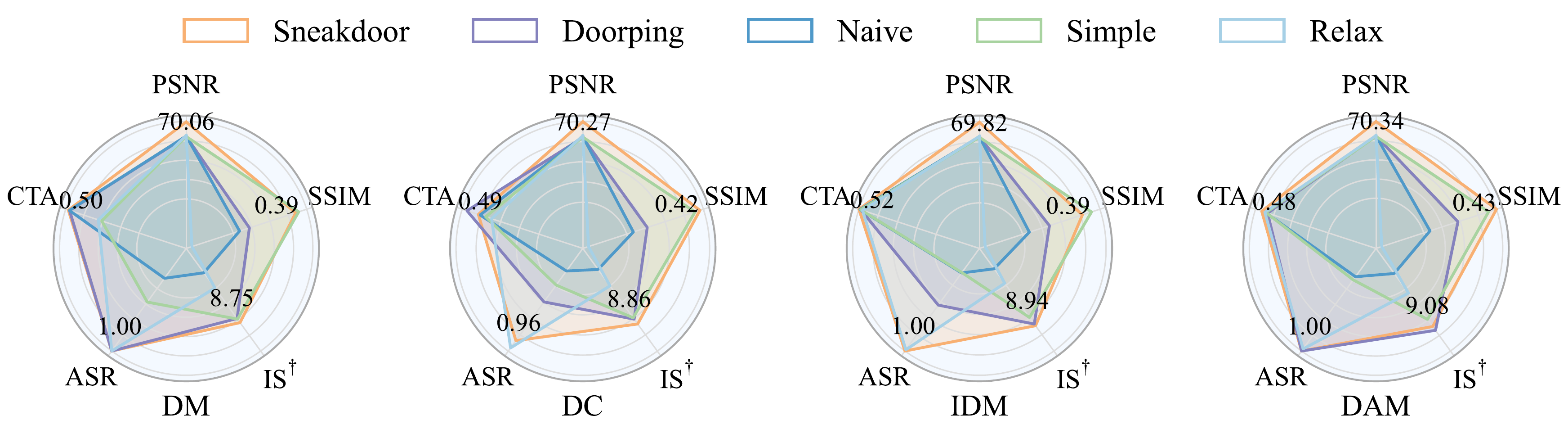}
    \caption{Attack Performance on Tiny-ImageNet. Larger area indicates better balance.}
    \label{fig:radar_balance_Tiny}
\end{figure}
\endgroup
\paragraph{Effectiveness on Different Datasets}

To rigorously assess the effectiveness of \textsc{sneakdoor}, we evaluate CTA and ASR across five datasets and four dataset condensation baselines: DM~\cite{zhao2023dm}, DC~\cite{zhao2020dataset}, IDM~\cite{zhao2023improved}, and DAM~\cite{sajedi2023datadam}. Results are summarized in Table~\ref{tab:cta_asr_on_convnet}, with each entry reporting the mean and standard deviation over five random seeds. \textsc{Sneakdoor} consistently achieves high ASR across all datasets and condensation methods, while maintaining competitive CTA. These results highlight the robustness and generalizability of \textsc{Sneakdoor}, with improvements most evident in scenarios where baseline methods overfit to specific condensation schemes.


\begingroup
\setlength{\tabcolsep}{3pt} 
\begin{table*}[htbp]
\centering
\caption{Effectiveness on Different Datasets}
\begin{adjustbox}{width=\textwidth}
{
\begin{tabular}{cc*{4}{c>{\columncolor{gray!25}}c}}
\toprule
\multirow{2}{*}{\textbf{Dataset}} & \multirow{2}{*}{\textbf{Method}} 
  & \multicolumn{2}{c}{\textbf{\textsc{Sneakdoor}}} 
  & \multicolumn{2}{c}{\textbf{DOORPING}} 
  & \multicolumn{2}{c}{\textbf{SIMPLE}} 
  & \multicolumn{2}{c}{\textbf{RELAX}} \\
& & CTA & ASR & CTA & ASR & CTA & ASR & CTA & ASR \\
\midrule
\multirow{4}{*}{\textbf{CIFAR10}}
& DM  & $0.626\pm0.001$ & $0.989\pm0.000$ 
       & $0.621\pm0.001$ & $0.988\pm0.005$ 
       & $0.584\pm0.000$ & $0.590\pm0.012$ 
       & $0.574\pm0.000$ & $1.000\pm0.000$ \\
& DC  & $0.537\pm0.000$ & $0.996\pm0.000$ 
       & $0.566\pm0.001$ & $1.000\pm0.000$  
       & $0.497\pm0.001$ & $0.657\pm0.021$ 
       & $0.511\pm0.001$ & $1.000\pm0.000$ \\
& IDM & $0.643\pm0.002$ & $0.975\pm0.001$ 
       & $0.654\pm0.002$ & $0.165\pm0.007$ 
       & $0.652\pm0.001$ & $0.142\pm0.008$ 
       & $0.653\pm0.002$ & $0.522\pm0.021$ \\
& DAM & $0.591\pm0.001$ & $0.979\pm0.001$ 
       & $0.531\pm0.001$ & $1.000\pm0.000$  
       & $0.537\pm0.001$ & $0.674\pm0.032$ 
       & $0.559\pm0.001$ & $1.000\pm0.001$ \\
\midrule
\multirow{4}{*}{\textbf{STL10}}
& DM  & $0.598\pm0.001$ & $0.973\pm0.000$ 
       & $0.577\pm0.001$ & $0.149\pm0.007$ 
       & $0.597\pm0.001$ & $0.096\pm0.009$ 
       & $0.596\pm0.001$ & $1.000\pm0.001$ \\
& DC  & $0.565\pm0.001$ & $0.998\pm0.001$ 
       & $0.598\pm0.001$ & $0.227\pm0.011$ 
       & $0.550\pm0.001$ & $0.112\pm0.011$ 
       & $0.563\pm0.000$ & $0.998\pm0.001$ \\
& IDM & $0.658\pm0.001$ & $0.979\pm0.001$ 
       & $0.661\pm0.001$ & $0.314\pm0.015$ 
       & $0.658\pm0.001$ & $0.100\pm0.007$ 
       & $0.658\pm0.001$ & $0.954\pm0.011$ \\
& DAM & $0.532\pm0.001$ & $0.992\pm0.001$ 
       & $0.533\pm0.001$ & $1.000\pm0.000$  
       & $0.535\pm0.001$ & $0.103\pm0.004$ 
       & $0.535\pm0.001$ & $1.000\pm0.000$ \\
\midrule
\multirow{4}{*}{\textbf{FMNIST}}
& DM  & $0.876\pm0.001$ & $0.998\pm0.000$ 
       & $0.876\pm0.000$ & $0.093\pm0.006$ 
       & $0.868\pm0.000$ & $0.178\pm0.005$ 
       & $0.828\pm0.000$ & $1.000\pm0.000$ \\
& DC  & $0.851\pm0.001$ & $0.998\pm0.000$ 
       & $0.872\pm0.001$ & $1.000\pm0.000$  
       & $0.837\pm0.001$ & $0.277\pm0.014$ 
       & $0.824\pm0.001$ & $1.000\pm0.000$ \\
& IDM & $0.877\pm0.001$ & $1.000\pm0.000$ 
       & $0.884\pm0.000$ & $0.998\pm0.002$ 
       & $0.879\pm0.000$ & $0.159\pm0.007$ 
       & $0.875\pm0.001$ & $1.000\pm0.000$ \\
& DAM & $0.877\pm0.000$ & $0.996\pm0.000$ 
       & $0.813\pm0.001$ & $1.000\pm0.000$  
       & $0.880\pm0.000$ & $0.151\pm0.012$ 
       & $0.874\pm0.000$ & $1.000\pm0.000$ \\
\midrule
\multirow{4}{*}{\textbf{SVHN}}
& DM  & $0.800\pm0.000$ & $1.000\pm0.000$ 
       & $0.780\pm0.001$ & $1.000\pm0.001$ 
       & $0.748\pm0.000$ & $0.110\pm0.007$ 
       & $0.747\pm0.000$ & $1.000\pm0.000$ \\
& DC  & $0.687\pm0.000$ & $1.000\pm0.000$ 
       & $0.583\pm0.001$ & $0.703\pm0.017$ 
       & $0.636\pm0.001$ & $0.100\pm0.009$ 
       & $0.689\pm0.001$ & $1.000\pm0.000$ \\
& IDM & $0.831\pm0.001$ & $0.986\pm0.001$ 
       & $0.839\pm0.001$ & $0.061\pm0.006$ 
       & $0.842\pm0.001$ & $0.114\pm0.008$ 
       & $0.834\pm0.002$ & $0.992\pm0.003$ \\
& DAM & $0.782\pm0.001$ & $1.000\pm0.000$ 
       & $0.721\pm0.000$ & $1.000\pm0.000$ 
       & $0.759\pm0.001$ & $0.114\pm0.005$ 
       & $0.745\pm0.001$ & $1.000\pm0.000$ \\
\midrule
\multirow{4}{*}{\shortstack{\textbf{TINY}\\\textbf{IMAGENET}}}
& DM  & $0.503\pm0.001$ & $1.000\pm0.000$ 
       & $0.496\pm0.002$ & $1.000\pm0.000$ 
       & $0.493\pm0.003$ & $0.100\pm0.004$ 
       & $0.494\pm0.003$ & $0.996\pm0.000$ \\
& DC  & $0.432\pm0.002$ & $1.000\pm0.000$ 
       & $0.492\pm0.001$ & $0.398\pm0.005$ 
       & $0.391\pm0.002$ & $0.192\pm0.006$ 
       & $0.418\pm0.003$ & $0.952\pm0.001$ \\
& IDM & $0.517\pm0.004$ & $1.000\pm0.000$ 
       & $0.512\pm0.005$ & $0.089\pm0.013$ 
       & $0.509\pm0.003$ & $0.046\pm0.002$ 
       & $0.484\pm0.006$ & $0.941\pm0.002$ \\
& DAM & $0.482\pm0.003$ & $1.000\pm0.000$ 
       & $0.449\pm0.003$ & $1.000\pm0.000$ 
       & $0.458\pm0.003$ & $0.082\pm0.002$ 
       & $0.465\pm0.002$ & $0.973\pm0.001$ \\
\bottomrule
\end{tabular}
}
\end{adjustbox}
\label{tab:cta_asr_on_convnet}
\end{table*}
\endgroup

\paragraph{Effectiveness on Cross Architectures}
To evaluate \textsc{Sneakdoor} in cross-architecture settings, where the condensation model differs from the downstream model, we follow prior work~\cite{liu2023backdoor} and consider four architectures: ConvNet, AlexNetBN, VGG11, and ResNet18. Specifically, we use ConvNet or AlexNetBN for data condensation and the remaining models for downstream training. 

As shown in Table~\ref{tab:cross_architecture_cta_asr}, we evaluate \textsc{Sneakdoor}. across 36 cross-architecture scenarios spanning various datasets, condensation methods, and downstream models. \textsc{Sneakdoor} demonstrates consistent performance across most architecture pairs, indicating strong transferability. However, when using the DC algorithm, performance systematically degrades on specific architectures. Prior studies, as well as our own findings, suggest that DC often produces lower-quality distilled datasets, as reflected in its relatively low CTA. This implies that the reduced ASR in these cases is more likely due to {\textit{ DC's limited ability to retain both task-relevant and backdoor-relevant information, rather than a shortcoming of the attack mechanism itself}}. When excluding DC-based cases, 27 scenarios remain, of which only 6 exhibit ASR below 90\%. This demonstrates that \textsc{Sneakdoor} consistently achieves high ASR in most settings, provided the underlying condensed data is of sufficient quality.


\begingroup
\setlength{\tabcolsep}{3pt}
\begin{table*}[ht]
\centering
\caption{Cross‐architecture CTA and ASR}
\begin{adjustbox}{width=\textwidth}
{
\begin{tabular}{cc*{4}{c>{\columncolor{gray!25}}c}}
\toprule
\multirow{2}{*}{\textbf{Dataset}} & \multirow{2}{*}{\textbf{Network}}
  & \multicolumn{2}{c}{\textbf{DM}}
  & \multicolumn{2}{c}{\textbf{DC}}
  & \multicolumn{2}{c}{\textbf{IDM}}
  & \multicolumn{2}{c}{\textbf{DAM}} \\
 & & CTA & ASR & CTA & ASR & CTA & ASR & CTA & ASR \\
\midrule
\multirow{3}{*}{\textbf{CIFAR10}}
 & VGG11     & $0.568\pm0.000$ & $0.971\pm0.000$ 
             & $0.472\pm0.000$ & $0.865\pm0.000$ 
             & $0.645\pm0.000$ & $0.719\pm0.008$ 
             & $0.539\pm0.000$ & $0.929\pm0.001$ \\
 & AlexNetBN & $0.616\pm0.001$ & $0.942\pm0.002$ 
             & $0.426\pm0.004$ & $0.000\pm0.000$ 
             & $0.689\pm0.002$ & $0.539\pm0.003$ 
             & $0.623\pm0.001$ & $0.902\pm0.004$ \\
 & ResNet18  & $0.548\pm0.001$ & $0.959\pm0.000$ 
             & $0.435\pm0.001$ & $0.534\pm0.003$ 
             & $0.656\pm0.001$ & $0.766\pm0.003$ 
             & $0.510\pm0.001$ & $0.857\pm0.002$ \\
\midrule
\multirow{3}{*}{\textbf{STL10}}
 & VGG11     & $0.587\pm0.001$ & $0.999\pm0.001$ 
             & $0.564\pm0.000$ & $0.790\pm0.003$ 
             & $0.676\pm0.001$ & $0.900\pm0.001$ 
             & $0.582\pm0.000$ & $0.924\pm0.001$ \\
 & AlexNetBN & $0.589\pm0.002$ & $0.905\pm0.005$ 
             & $0.542\pm0.001$ & $0.796\pm0.002$ 
             & $0.670\pm0.003$ & $0.798\pm0.005$ 
             & $0.636\pm0.001$ & $0.981\pm0.001$ \\
 & ResNet18  & $0.463\pm0.001$ & $0.989\pm0.000$ 
             & $0.396\pm0.001$ & $0.783\pm0.003$ 
             & $0.647\pm0.001$ & $0.949\pm0.001$ 
             & $0.436\pm0.001$ & $0.941\pm0.002$ \\
\midrule
\multirow{3}{*}{\shortstack{\textbf{TINY}\\\textbf{IMAGENET}}}
 & VGG11     & $0.488\pm0.001$ & $1.000\pm0.000$ 
             & $0.384\pm0.001$ & $1.000\pm0.000$ 
             & $0.541\pm0.002$ & $1.000\pm0.000$ 
             & $0.449\pm0.002$ & $1.000\pm0.000$ \\
 & AlexNetBN & $0.517\pm0.003$ & $0.796\pm0.015$ 
             & $0.292\pm0.007$ & $0.704\pm0.008$ 
             & $0.572\pm0.004$ & $1.000\pm0.000$ 
             & $0.541\pm0.003$ & $1.000\pm0.000$ \\
 & ResNet18  & $0.456\pm0.002$ & $1.000\pm0.000$ 
             & $0.358\pm0.001$ & $0.524\pm0.008$ 
             & $0.483\pm0.005$ & $0.988\pm0.010$ 
             & $0.438\pm0.002$ & $1.000\pm0.000$ \\
\bottomrule
\end{tabular}
}
\end{adjustbox}
\label{tab:cross_architecture_cta_asr}
\end{table*}
\endgroup

\paragraph{Evaluation of Stealthiness}

As shown in Figure~\ref{fig:bar_stl10}, \textsc{Sneakdoor} consistently achieves the highest PSNR and SSIM across all condensation methods, highlighting its ability to produce visually and structurally imperceptible triggers. In contrast, the other methods exhibit notable declines in both metrics, suggesting visible artifacts or structural distortions in the perturbed samples. Moreover, while Simple and Naive achieve slightly lower IS values, they fail to maintain competitive ASR or CTA, limiting their overall effectiveness. \textsc{Sneakdoor} achieves a similarly low IS while preserving high ASR, indicating enhanced stealth without sacrificing attack strength.


\begingroup
\begin{figure}[h]
    \centering
    \includegraphics[width=0.9\textwidth]{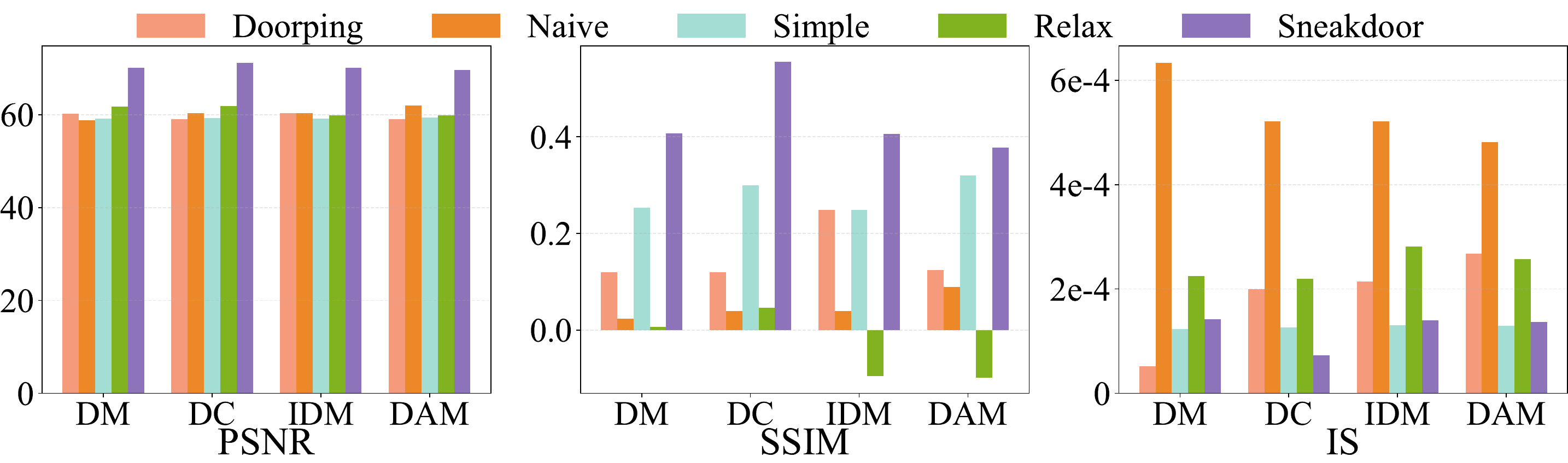}
    \caption{Stealthiness Performance on STL10}
    \label{fig:bar_stl10}
\end{figure}
\endgroup
\paragraph{Robust to Defense}
To evaluate the resilience of \textsc{Sneakdoor} against existing defense mechanisms, we conduct comprehensive experiments spanning model-level, input-level, and dataset-level defenses. Results in Table~\ref{tab:combined_nc_reasr_pixel} show that \textsc{Sneakdoor} consistently evades state-of-the-art model-level defenses such as NC~\cite{nc} and PIXEL~\cite{taog2022better}, with all anomaly scores remaining below detection thresholds. Input-level defenses also fail to recover effective triggers, as indicated by uniformly low REASR values across all settings~\cite{abs}. While dataset-level methods such as RNP~\cite{li2023reconstructive} and PDB~\cite{wei2024mitigating} succeed in suppressing ASR, they face significant drops in CTA, reflecting a sharp trade-off. These findings highlight \textsc{Sneakdoor} as a robust attack that remains effective under diverse defense conditions.


\begingroup
\setlength{\tabcolsep}{3pt}
\begin{table*}[htbp]
\centering
\caption{NC, ABS, and PIXEL across different datasets and condensation methods.}
\begin{adjustbox}{width=\textwidth}
\begin{tabular}{%
  c
  >{\columncolor{gray!25}}c c
  >{\columncolor{gray!25}}c c
  >{\columncolor{gray!25}}c c
  >{\columncolor{gray!25}}c c
  >{\columncolor{gray!25}}c c
  >{\columncolor{gray!25}}c c
}
\toprule
\multirow{2}{*}{\textbf{Dataset}}
  & \multicolumn{4}{c}{\textbf{NC Anomaly Index}}
  & \multicolumn{4}{c}{\textbf{ABS REASR}}
  & \multicolumn{4}{c}{\textbf{PIXEL}} \\
 & 
  \multicolumn{1}{>{\columncolor{white}}c}{\textbf{DM}} & 
  \multicolumn{1}{>{\columncolor{white}}c}{\textbf{DC}} & 
  \multicolumn{1}{>{\columncolor{white}}c}{\textbf{IDM}} & 
  \multicolumn{1}{>{\columncolor{white}}c}{\textbf{DAM}} &
  \multicolumn{1}{>{\columncolor{white}}c}{\textbf{DM}} & 
  \multicolumn{1}{>{\columncolor{white}}c}{\textbf{DC}} & 
  \multicolumn{1}{>{\columncolor{white}}c}{\textbf{IDM}} & 
  \multicolumn{1}{>{\columncolor{white}}c}{\textbf{DAM}} &
  \multicolumn{1}{>{\columncolor{white}}c}{\textbf{DM}} & 
  \multicolumn{1}{>{\columncolor{white}}c}{\textbf{DC}} & 
  \multicolumn{1}{>{\columncolor{white}}c}{\textbf{IDM}} & 
  \multicolumn{1}{>{\columncolor{white}}c}{\textbf{DAM}} \\
\midrule
STL10         & 1.3180 & 1.0872 & 1.3648 & 0.9843 & 0.19 & 0.19 & 0.25 & 0.17 & 1.5525 & 1.0515 & 0.7688 & 1.5425 \\
CIFAR10       & 1.8762 & 0.9518 & 1.7640 & 1.3787 & 0.24 & 0.35 & 0.29 & 0.57 & 1.7705 & 1.2625 & 1.7750 & 0.9472 \\
TINY-IMAGENET & 1.4706 & 1.6199 & 1.2201 & 1.9065 & 0.17 & 0.14 & 0.15 & 0.16 & 1.7813 & 1.4252 & 1.9528 & 1.3447 \\
\bottomrule
\end{tabular}
\end{adjustbox}
\label{tab:combined_nc_reasr_pixel}
\end{table*}
\endgroup

\begingroup

\begin{table}[!htbp]
\centering
\setlength{\tabcolsep}{13pt} 
\caption{Effects of (1) Class Pair Selection and (2) Input-Aware Trigger Generation}
\begin{adjustbox}{width=\textwidth,keepaspectratio}
\begin{tabular}{ccccccc}
\toprule
\textbf{(1)} & \textbf{(2)}   & \textbf{CTA}           & \textbf{ASR}           & \textbf{PSNR} & \textbf{SSIM} & \textbf{IS}               \\
\midrule
\XSolidBrush 
 & \Checkmark & $0.5912\pm0.0004$ & $0.9946\pm0.0005$ & 65.8677 & 0.12915 & $1.3058\times10^{-5}$ \\
\midrule
 \Checkmark & \XSolidBrush  & $0.6211\pm0.0005$ & $0.9876\pm0.0050$ & 59.8469 & 0.08217 & $2.2987\times10^{-4}$ \\
 \Checkmark & \Checkmark & $0.6262\pm0.0005$ & $0.9890\pm0.0000$ & 73.2285 & 0.66151 & $4.8441\times10^{-5}$ \\
\bottomrule
\end{tabular}
\end{adjustbox}
\label{tab:ablation}
\end{table}

\begingroup
\begin{table*}[!htbp]
\centering
\setlength{\tabcolsep}{12pt} 
\caption{CTA/ASR Before and After Defense}
\begin{adjustbox}{width=\textwidth,keepaspectratio}
\begin{tabular}{cccccc}
\toprule
\textbf{Dataset} & \textbf{Method} & \textbf{DM} & \textbf{DC} & \textbf{DAM} & \textbf{IDM} \\
\midrule

\multirow{3}{*}{CIFAR10}
 & W/O Defense & 0.6262/0.9890 & 0.5372/0.9960 & 0.5906/0.9794 & 0.6431/0.9754 \\
 & RNP    & 0.2334/0.5490 & 0.3874/0.1340 & 0.5748/0.9850 & 0.4424/0.2870 \\
 & PDB    & 0.1388/0.1380 & 0.1000/0.0000 & 0.0664/0.0300 & 0.3191/0.4190 \\
\midrule

\multirow{3}{*}{STL10}
 & W/O Defense & 0.5979/0.9725 & 0.5653/0.9975 & 0.5324/0.9918 & 0.6582/0.9790 \\
 & RNP    & 0.2791/0.0625 & 0.3955/0.8962 & 0.4961/0.8488 & 0.4889/0.5887 \\
 & PDB    & 0.4719/0.0425 & 0.1150/0.0100 & 0.1293/0.0313 & 0.2646/0.0038 \\
\midrule


\multirow{3}{*}{TINY-IMAGENET}
 & W/O Defense & 0.5026/1.0000 & 0.4318/1.0000 & 0.4822/1.0000 & 0.5174/1.0000 \\
 & RNP    & 0.2700/0.0600 & 0.2450/0.0200 & 0.3320/0.7600 & 0.3450/0.9200 \\
 & PDB    & 0.1030/0.0000 & 0.0570/0.0000 & 0.0540/0.0000 & 0.0800/0.1600 \\

\bottomrule
\end{tabular}
\end{adjustbox}
\label{tab:combined_pdb_rnp}
\end{table*}
\endgroup

\paragraph{Ablation study}
To assess the contribution of key components in \textsc{Sneakdoor}, we perform ablation studies on \textit{(1) inter-class boundary-based class pair selection and (2) input-aware trigger generation}. Removing (1) and using arbitrary class pairs slightly reduces ASR but significantly degrades CTA and stealth metrics (PSNR, SSIM). Replacing (2) with fixed patterns, as in Doorping, maintains ASR and CTA but severely compromises stealthiness, as shown by reduced similarity and elevated IS. These results underscore the necessity of both components.

Due to space limitations, we report supplementary results in Appendix C, including comparisons with additional attack baselines, analysis of varying the number of condensed samples per class, and evaluations using AlexNet as the condensation model.

\endgroup

\section{Limitations}

While \textsc{Sneakdoor} achieves a good balance across ASR, CTA, and STE, it does not consistently surpass all existing methods on any single metric. In certain cases, baseline approaches such as DOORPING attain higher ASR or CTA when considered in isolation. This trade-off reflects the inherent challenge of jointly optimizing multiple, often competing objectives. Future work could investigate methods that enhance a specific metric without sacrificing other metrics. Further refinement may lead to more adaptable backdoor attacks tailored to specific deployment or threat scenarios. Another limitation lies in the dependence on a relatively high poisoning ratio to reach optimal attack effectiveness. Reducing this requirement would make the approach more practical in real-world scenarios where the attacker’s control over data is limited. Finally, \textsc{Sneakdoor} does not fully capture more complex threat models that involve targeted source-to-target manipulations, such as altering ``Stop Sign'' to ``Speed Limit: 60 mph'', which poses serious safety risks. In such cases, the attack’s effectiveness may decrease. Extending \textsc{Sneakdoor} to handle diverse and task-specific attack objectives remains an important direction for future research.

\section{Conclusion}
This work introduces \textsc{Sneakdoor}, a novel attack paradigm that exposes critical vulnerabilities in distribution-matching–based dataset condensation methods. By integrating input-aware trigger generation with inter-class misclassification analysis, \textsc{Sneakdoor} injects imperceptible yet highly effective backdoors into synthetic datasets. The theoretical analysis in reproducing kernel Hilbert space (RKHS) formalizes the stealth properties of the attack, showing that the induced perturbations remain bounded in both geometric and distributional space. Extensive experiments across multiple datasets, condensation baselines, and defense strategies confirm that \textsc{Sneakdoor} achieves strong ASR–CTA–STE trade-offs and maintains high transferability under cross-architecture evaluation. Together, these results reveal that even condensed data, often regarded as a privacy-preserving substitute for raw data, can serve as a potent vector for model compromise when the condensation process is adversarially controlled. {\textit{This study lays the foundation for understanding the vulnerabilities and defense limitations of current condensation frameworks}}, emphasizing the need for proactive safeguards in synthetic data pipelines. 


\section*{Broader Impact}
Backdoor attacks against dataset condensation pose significant risks given the growing use of condensed datasets in privacy-sensitive or resource-constrained settings such as outsourced data compression, federated learning, machine unlearning, and continual learning. For instance, in continual learning systems deployed in edge AI applications, such as autonomous vehicles or medical diagnosis assistants, lightweight condensed datasets enable efficient model updates without full retraining. If an adversary injects imperceptible backdoor triggers into this data, the resulting models may misclassify critical inputs (\textit{e.g}., road signs or tumor types), leading to serious safety and ethical consequences. Given these risks, the responsible disclosure of such attacks is essential. The goal of our work is to expose vulnerabilities in distribution-matching-based condensation methods to inform the design of more effective defenses. To mitigate misuse, we recommend: (1) incorporating robust anomaly detection and certified defenses during condensation; (2) encouraging transparency and reproducibility in condensation pipelines; and (3) enforcing rigorous provenance tracking to dataset generation processes. Our findings serve both as a cautionary signal and a foundation for developing secure and resilient dataset condensation techniques.

\section*{Acknowledgments}
This work was supported by the National Key Research and Development Program of China (No.2023YFF0905300), NSFC Grant. NO.62176205 and 62472346.

{
\small
\bibliographystyle{unsrt} \bibliography{neurips}
}

\newpage
\section*{NeurIPS Paper Checklist}

The checklist is designed to encourage best practices for responsible machine learning research, addressing issues of reproducibility, transparency, research ethics, and societal impact. Do not remove the checklist: {\bf The papers not including the checklist will be desk rejected.} The checklist should follow the references and follow the (optional) supplemental material.  The checklist does NOT count towards the page
limit. 

Please read the checklist guidelines carefully for information on how to answer these questions. For each question in the checklist:
\begin{itemize}
    \item You should answer \answerYes{}, \answerNo{}, or \answerNA{}.
    \item \answerNA{} means either that the question is Not Applicable for that particular paper or the relevant information is Not Available.
    \item Please provide a short (1–2 sentence) justification right after your answer (even for NA). 
\end{itemize}

{\bf The checklist answers are an integral part of your paper submission.} They are visible to the reviewers, area chairs, senior area chairs, and ethics reviewers. You will be asked to also include it (after eventual revisions) with the final version of your paper, and its final version will be published with the paper.

The reviewers of your paper will be asked to use the checklist as one of the factors in their evaluation. While "\answerYes{}" is generally preferable to "\answerNo{}", it is perfectly acceptable to answer "\answerNo{}" provided a proper justification is given (e.g., "error bars are not reported because it would be too computationally expensive" or "we were unable to find the license for the dataset we used"). In general, answering "\answerNo{}" or "\answerNA{}" is not grounds for rejection. While the questions are phrased in a binary way, we acknowledge that the true answer is often more nuanced, so please just use your best judgment and write a justification to elaborate. All supporting evidence can appear either in the main paper or the supplemental material, provided in appendix. If you answer \answerYes{} to a question, in the justification please point to the section(s) where related material for the question can be found.

IMPORTANT, please:
\begin{itemize}
    \item {\bf Delete this instruction block, but keep the section heading ``NeurIPS Paper Checklist"},
    \item  {\bf Keep the checklist subsection headings, questions/answers and guidelines below.}
    \item {\bf Do not modify the questions and only use the provided macros for your answers}.
\end{itemize}


\begin{enumerate}

\item {\bf Claims}
    \item[] Question: Do the main claims made in the abstract and introduction accurately reflect the paper's contributions and scope?
    \item[] Answer: \answerYes{} 
    \item[] Justification: The abstract and introduction outline the motivation and detail the technical contributions of the proposed approach.

\item {\bf Limitations}
    \item[] Question: Does the paper discuss the limitations of the work performed by the authors?
    \item[] Answer: \answerYes{} 
    \item[] Justification: While \textsc{Sneakdoor} achieves the best overall balance across Attack Success Rate (ASR), Clean Test Accuracy (CTA), and Stealthiness (STE), it does not consistently outperform existing methods on any single metric.

\item {\bf Theory assumptions and proofs}
    \item[] Question: For each theoretical result, does the paper provide the full set of assumptions and a complete (and correct) proof?
    \item[] Answer: \answerYes{} 
    \item[] Justification: Each theoretical result is provided the full set of assumptions and a complete (and correct) proof.

    \item {\bf Experimental result reproducibility}
    \item[] Question: Does the paper fully disclose all the information needed to reproduce the main experimental results of the paper to the extent that it affects the main claims and/or conclusions of the paper (regardless of whether the code and data are provided or not)?
    \item[] Answer: \answerYes{} 
    \item[] Justification: The disclosed information is enough to reproduce the main experiments. We will also release the source code late.

\item {\bf Open access to data and code}
    \item[] Question: Does the paper provide open access to the data and code, with sufficient instructions to faithfully reproduce the main experimental results, as described in supplemental material?
    \item[] Answer: \answerYes{} 
    \item[] Justification: We provide essential parts for the code and details in supplemental material.

\item {\bf Experimental setting/details}
    \item[] Question: Does the paper specify all the training and test details (e.g., data splits, hyperparameters, how they were chosen, type of optimizer, etc.) necessary to understand the results?
    \item[] Answer: \answerYes{} 
    \item[] Justification: We provide all details about the experiments.

\item {\bf Experiment statistical significance}
    \item[] Question: Does the paper report error bars suitably and correctly defined or other appropriate information about the statistical significance of the experiments?
    \item[] Answer: \answerYes{} 
    \item[] Justification: The error is shown in our experiments.

\item {\bf Experiments compute resources}
    \item[] Question: For each experiment, does the paper provide sufficient information on the computer resources (type of compute workers, memory, time of execution) needed to reproduce the experiments?
    \item[] Answer: \answerYes{} 
    \item[] Justification: All experiments were conducted utilizing the NVIDIA GeForce RTX 4090 GPU.
    
\item {\bf Code of ethics}
    \item[] Question: Does the research conducted in the paper conform, in every respect, with the NeurIPS Code of Ethics \url{https://neurips.cc/public/EthicsGuidelines}?
    \item[] Answer: \answerYes{} 
    \item[] Justification: We follow the policy.

\item {\bf Broader impacts}
    \item[] Question: Does the paper discuss both potential positive societal impacts and negative societal impacts of the work performed?
    \item[] Answer: \answerYes{} 
    \item[] Justification: Backdoor attacks against dataset condensation pose significant risks given the growing use of condensed datasets in privacy-sensitive or resource-constrained settings such as outsourced data compression, federated learning, machine unlearning, and continual learning. To mitigate misuse, we recommend: (1) incorporating robust anomaly detection and certified defenses during condensation; (2) encouraging transparency and reproducibility in condensation pipelines; and (3) enforcing rigorous provenance tracking to dataset generation processes. 
    
\item {\bf Safeguards}
    \item[] Question: Does the paper describe safeguards that have been put in place for responsible release of data or models that have a high risk for misuse (e.g., pretrained language models, image generators, or scraped datasets)?
    \item[] Answer: \answerNo{} 
    \item[] Justification: The primary contribution of our proposed \textsc{Sneakdoor} is to expose vulnerabilities in distribution-matching-based condensation methods. Our work lays the groundwork for understanding the attack surface and limitations of current defenses, enabling the community to proactively build secure and trustworthy dataset condensation frameworks.
    \item[] Guidelines:
    \begin{itemize}
        \item The answer NA means that the paper poses no such risks.
        \item Released models that have a high risk for misuse or dual-use should be released with necessary safeguards to allow for controlled use of the model, for example by requiring that users adhere to usage guidelines or restrictions to access the model or implementing safety filters. 
        \item Datasets that have been scraped from the Internet could pose safety risks. The authors should describe how they avoided releasing unsafe images.
        \item We recognize that providing effective safeguards is challenging, and many papers do not require this, but we encourage authors to take this into account and make a best-faith effort.
    \end{itemize}

\item {\bf Licenses for existing assets}
    \item[] Question: Are the creators or original owners of assets (e.g., code, data, models), used in the paper, properly credited and are the license and terms of use explicitly mentioned and properly respected?
    \item[] Answer: \answerYes{} 
    \item[] Justification: We cite the original paper that produced the code package or dataset.

\item {\bf New assets}
    \item[] Question: Are new assets introduced in the paper well documented and is the documentation provided alongside the assets?
    \item[] Answer: \answerNA{} 
    \item[] Justification: \answerNA{}

\item {\bf Crowdsourcing and research with human subjects}
    \item[] Question: For crowdsourcing experiments and research with human subjects, does the paper include the full text of instructions given to participants and screenshots, if applicable, as well as details about compensation (if any)? 
    \item[] Answer: \answerNA{} 
    \item[] Justification: \answerNA{}

\item {\bf Institutional review board (IRB) approvals or equivalent for research with human subjects}
    \item[] Question: Does the paper describe potential risks incurred by study participants, whether such risks were disclosed to the subjects, and whether Institutional Review Board (IRB) approvals (or an equivalent approval/review based on the requirements of your country or institution) were obtained?
    \item[] Answer: \answerNA{} 
    \item[] Justification: \answerNA{}

\item {\bf Declaration of LLM usage}
    \item[] Question: Does the paper describe the usage of LLMs if it is an important, original, or non-standard component of the core methods in this research? Note that if the LLM is used only for writing, editing, or formatting purposes and does not impact the core methodology, scientific rigorousness, or originality of the research, declaration is not required.
    \item[] Answer: \answerNA{} 
    \item[] Justification: The LLM was used solely for language editing and clarity improvement. It did not contribute to the design, implementation, or validation of the proposed methods.
\end{enumerate}

\newpage
\appendix

\section{Attacker's Goal}
\textbf{Attacker’s Goal.} The attacker aims to achieve a multi-faceted objective when injecting backdoors into condensed datasets. This objective consists of three key goals: maintaining stealthiness, ensuring backdoor effectiveness, and preserving model utility on clean data.

\textit{Stealthiness (STE)}. The attacker’s goal is to ensure that malicious modifications remain imperceptible. This involves two requirements. Firstly, the poisoned condensed dataset $\widetilde{\mathcal{D}}$ must be visually and statistically indistinguishable from the clean version $\mathcal{D}$. This is critical, as condensed datasets are small ($|\widetilde{\mathcal{D}}|\ll|\mathcal{D}|$) and likely to be examined manually. Secondly, the triggered test samples remain imperceptibly different from unmodified test data. This requirement ensures that the backdoor remains undetectable during evaluation or deployment, whether through human inspection or automated analysis. 

\textit{Attack Success Rate (ASR)}. In parallel, the attacker aims to embed a functional backdoor that remains inactive during standard operation but activates reliably in the presence of a specific trigger. Let $f$ denote the downstream model trained on $\widetilde{\mathcal{D}}$ and $\Delta$ the backdoor trigger. For a triggered test sample $x_i+\Delta$, the ASR defined as:
\begin{equation}
    ASR=\frac{1}{N_t}\sum_{i=1}^{N_t}\mathbb{I}(f(x_i+\Delta)=t)
\end{equation}
where $t$ is the target label, $N_t$ is the number of triggered test samples, and $\mathbb{I}$ is the indicator function. The attacker aims to maximize ASR.

\textit{Clean Test Accuracy (CTA)}. Simultaneously, the attacker must preserve model accuracy on clean, non-triggered data. In other words, the condensed dataset must retain sufficient utility to support standard training objectives. This ensures that models trained on the poisoned data still generalize well to benign test sets. Let the clean test accuracy be defined as:
\begin{equation}
    CTA=\frac{1}{N_c}\sum_{i=1}^{N_c}\mathbb{I}(f(x_i)=y_i)
\end{equation}
where $y_i$ is the ground truth label of the test sample $x_i$, $N_c$ is the number of clean test samples. The attacker seeks to maintain a high CTA so that the backdoor remains covert.

\section{Stealthiness Analysis}
A critical challenge in designing effective backdoor attacks on dataset condensation is achieving stealthiness, ensuring that poisoned samples and the resulting synthetic data are indistinguishable from their clean counterparts. Our goal is to formalize stealthiness through a geometric and distributional lens, grounded in the feature space induced by deep neural architectures.

To this end, our analysis is guided by the following question: How does input-aware backdoor injection perturb the structure of data manifolds in feature space, and can this deviation be rigorously bounded to guarantee stealth? Since distribution matching-based condensation aligns global feature statistics (\textit{e.g.}, moments of embedded data), it is essential to understand whether triggers introduce detectable geometric or statistical anomalies in the condensed representation. We conduct our analysis in a Reproducing Kernel Hilbert Space (RKHS)~\cite{aronszajn1950theory,berlinet2011reproducing,ghojogh2021reproducing}, where class-specific data, both clean and triggered, are assumed to lie on smooth, locally compact manifolds. By modeling the trigger as a bounded, input-aware perturbation and invoking assumptions on manifold regularity and inter-class proximity, we show that triggered samples remain tightly coupled to the clean data manifold under mild conditions. This theoretical framework enables us to quantify the effect of poisoning both at the feature level (Theorem~\ref{theorem:upper bound on feature-manifold}) and at the level of the condensed dataset (Theorem~\ref{theorem:MMD}). These results provide principled justification for \textsc{Sneakdoor}’s empirical stealth: the perturbations introduced by the trigger remain latent-space-aligned and distributionally consistent, limiting their detectability after condensation.

\begin{assumption}[Lipschitz Continuity]
    The feature mapping $f_{\theta_f} : \mathcal{X} \to \mathcal{H}$ is assumed to be Lipschitz continuous. That is, for all $x, x' \in \mathcal{X}$,
    \begin{equation}
        \|f_{\theta_f}(x) - f_{\theta_f}(x')\|_{\mathcal{H}} \leq L_f \|x - x'\|_{\infty},
    \end{equation}
    
    where $L_f \in \mathbb{R}^+$ denotes the Lipschitz constant, and $\|\cdot\|_{\infty}$ is the $L_\infty$-norm in the input space.
\label{assumption:Lipschitz Continuity}
\end{assumption}

\begin{assumption}[Local Compactness of Feature Manifolds]
Let the clean target class dataset $\mathcal{T}_{y_\tau}$ and the triggered dataset $\mathcal{T}_\text{triggered}$ lie on smooth manifolds $\mathcal{M}_\text{clean}$ and $\mathcal{M}_\text{triggered}$, respectively, embedded in a Reproducing Kernel Hilbert Space (RKHS) $\mathcal{H}$. The following condition holds: For any point $z \in \mathcal{M}_\text{clean}$, there exists a neighborhood $\mathcal{N}(z) \subset \mathcal{H}$ and a diffeomorphism $\varphi_z: \mathcal{N}(z) \cap \mathcal{M}_\text{clean} \to U \subset \mathbb{R}^d$, where $U$ is an open subset and $d$ is the intrinsic dimension of the manifold. 
\label{assumption:local compactness}
\end{assumption}

\begin{assumption}[Inter-Class Hausdorff Distance]
Let $\mathcal{M}_\text{source}$ and $\mathcal{M}_\text{clean}$ denote the RKHS-embedded manifolds of the source and target (clean) classes, respectively. Their Hausdorff distance is defined as:
\begin{equation}
\delta\triangleq\sup_{z_s\in\mathcal{M}_\text{source}}\inf_{z_\tau \in \mathcal{M}_\text{clean}}\|z_s - z_\tau\|_{\mathcal{H}}
\end{equation}

This condition implies that the decision boundary between source and target classes is locally reachable in feature space, enabling feasible cross-class perturbations by the trigger generator.
\label{assumption:Inter-Class Boundary Proximi}
\end{assumption}

\begin{lemma}[Boundedness of Latent Space Perturbation]
    Under Assumption~\ref{assumption:Lipschitz Continuity} (Lipschitz Continuity), the perturbation in the latent space of the triggered sample $\widetilde{x} = x + \alpha G_\phi(x)$ is bounded as follows:
    \begin{equation}
    \|f_{\theta_f}(\widetilde{x}) - f_{\theta_f}(x)\|_{\mathcal{H}} \leq L_f\alpha\varepsilon,
    \end{equation}
    
    where $L_f$ is the Lipschitz constant of the feature mapping $f_{\theta_f}$, and $\varepsilon$ is the upper bound on the input perturbation, satisfying $\|G_\phi(x)\|_\infty \leq \varepsilon$.
\label{lemma:Boundedness of Latent Space Perturbatio}
\end{lemma}
\begin{proof}
    According to Eq (\ref{eq:gen}), the perturbation generated by the trigger generator $G_\phi$  satisfies the input space constraint $\|G_\phi(x)\|_\infty\leq\varepsilon$. Therefore, the following conclusion can be obtained:
    \begin{equation}
    \begin{aligned}
        \|f_{\theta_f}(\widetilde{x}) - f_{\theta_f}(x)\|_{\mathcal{H}}&=\|f_{\theta_f}(x+\alpha G_\phi(x)) - f_{\theta_f}(x)\|_{\mathcal{H}}\\
        &\leq L_f\|\alpha G_\phi(x)\|_\infty\\
        &\leq L_f\alpha\varepsilon
    \end{aligned}
    \end{equation}

This lemma shows that the perturbation's effect in the feature space is controlled by both the input perturbation bound $\alpha$, $\varepsilon$ and the Lipschitz constant $L_f$.
\end{proof}
\begin{lemma}
    Let $\mathcal{M}_\text{clean}$ and $\mathcal{M}_\text{triggered}$ be smooth manifolds in the Reproducing Kernel Hilbert Space (RKHS) $\mathcal{H}$, induced by the feature map $f_{\theta_f}:\mathcal{X}\mapsto\mathcal{H}$. Under Assumption~\ref{assumption:Lipschitz Continuity}, \ref{assumption:local compactness}, and \ref{assumption:Inter-Class Boundary Proximi}, there exists a diffeomorphism $\Psi: \mathcal{M}_\text{source} \to \mathcal{M}_\text{triggered}$ such that: (1) $\sup_{z_s \in \mathcal{M}_\text{source}} \| \Psi(z_s) - z_s \|_{\mathcal{H}} \leq \gamma \varepsilon, \quad \text{where } \gamma = L_f \alpha$. (2) $\mathcal{M}_\text{triggered} \subset \mathcal{N}_{\delta'}(\mathcal{M}_\text{clean}), \quad \delta' = L_f \alpha \varepsilon + \delta$, where $\mathcal{N}_{\delta'}(\mathcal{M}_\text{clean})$ denotes the $\delta'$-neighborhood of $\mathcal{M}_\text{clean}$ in $\mathcal{H}$.
\label{lemma:diffeomorphism}
\end{lemma}
\begin{proof}
    By Assumption~\ref{assumption:local compactness}, for each $z_s \in \mathcal{M}_\text{source}$, there exists a local chart $\varphi_s: \mathcal{N}(z_s) \cap \mathcal{M}_\text{source} \to U_s \subset \mathbb{R}^d$, where $\mathcal{N}(z_s) \subset \mathcal{H}$ is a neighborhood and $U_s$ is an open subset.

    Define the local mapping $\psi_s: U_s \mapsto \mathcal{M}_\text{triggered}$ by:
    \begin{equation}
        \psi_s(u) = f_{\theta_f}\left( f_{\theta_f}^{-1}(\varphi_s^{-1}(u)) + \alpha G_\phi(f_{\theta_f}^{-1}(\varphi_s^{-1}(u))) \right)
    \end{equation}
    
    The smoothness of $\psi_s$ follows from the differentiability of $G_\phi$ and $f_{\theta_f}$. Then, by Lemma~\ref{lemma:Boundedness of Latent Space Perturbatio}, we can obtain: $\|\psi_s(u) - \varphi_s^{-1}(u)\|_{\mathcal{H}} \leq L_f \alpha \varepsilon = \gamma \varepsilon$.

    To construct a global diffeomorphism, take a finite open cover $\{ \mathcal{N}(z_{s_i}) \}_{i=1}^k$ of $\mathcal{M}_\text{source}$, with corresponding charts $\varphi_{s_i}$ and a smooth partition of unity $\{ \rho_i \}$:
    \begin{equation}
        \Psi(z_s) = \sum_{i=1}^k \rho_i(z_s) \cdot \psi_{s_i}(\varphi_{s_i}(z_s)).
    \end{equation}

    We now bound the total perturbation:
    \begin{equation}
    \begin{aligned}
    \| \Psi(z_s) - z_s \|_{\mathcal{H}} 
    &\leq \sum_{i=1}^k \rho_i(z_s)\| \psi_{s_i}(\varphi_{s_i}(z_s)) - z_s \|_{\mathcal{H}} \\
    &\leq \sum_{i=1}^k \rho_i(z_s) L_f \alpha \varepsilon \\
    &= L_f \alpha \varepsilon \\
    &= \gamma \varepsilon
    \end{aligned}
    \end{equation}

For any $z_t \in \mathcal{M}_\text{triggered}$, there exists $z_s \in \mathcal{M}_\text{source}$ such that $z_t = \Psi(z_s)$. By Assumption 3, there exists $z_\tau \in \mathcal{M}_\text{clean}$ with
$\|z_s - z_\tau\|_{\mathcal{H}} \leq \delta$. Then by the triangle inequality:

\begin{equation}
\begin{aligned}
    \|z_t - z_\tau\|_{\mathcal{H}} &\leq \|z_t - z_s\|_{\mathcal{H}} + \|z_s - z_\tau\|_{\mathcal{H}} \\
    &\leq L_f \alpha \varepsilon + \delta = \delta'
\end{aligned}
\end{equation}

Hence, $\mathcal{M}_\text{triggered} \subset \mathcal{N}_{\delta'}(\mathcal{M}_\text{clean})$.

To verify $\Psi$ is a diffeomorphism:
\begin{itemize}
    \item {Injectivity:} Follows from local injectivity of each $\psi_{s_i}$ and the partition of unity.
    \item {Surjectivity:} For any $z_t \in \mathcal{M}_\text{triggered}$, there exists $x \in \mathcal{T}_{y_s}$ such that $z_t = f_{\theta_f}(x + \alpha G_\phi(x)) = \Psi(f_{\theta_f}(x))$.
    \item {Smooth Inverse:} Local inverses $\psi_{s_i}^{-1}$ exist by the inverse function theorem and can be smoothly blended via $\{\rho_i\}$.
\end{itemize}
\end{proof}

\begin{theorem}[Upper Bound on Feature-Manifold Deviation under Poisoning]
Let \( \mathcal{T}_{y_\tau} \) denote the clean target-class dataset and \( \mathcal{T}_{\mathrm{triggered}} \) the triggered (poisoned) dataset, with corresponding feature-space distributions \( P_{\mathcal{M}_{\mathrm{clean}}} \) and \( P_{\mathcal{M}_{\mathrm{triggered}}} \), respectively. Define the mixed distribution as:
\[
P_{\mathcal{M}_{\mathrm{mixed}}} = (1-\rho) P_{\mathcal{M}_{\mathrm{clean}}} + \rho P_{\mathcal{M}_{\mathrm{triggered}}},
\]

where \( \rho \in [0,1] \) denotes the poisoning ratio. Under Assumptions \ref{assumption:Lipschitz Continuity}, \ref{assumption:local compactness}, and \ref{assumption:Inter-Class Boundary Proximi}, the expected deviation of samples from the mixed distribution to the target feature manifold satisfies:
\begin{equation}
    \mathbb{E}_{z \sim P_{\mathcal{M}_{\mathrm{mixed}}}} \left[ \inf_{z_\tau \in \mathcal{M}_{\mathrm{clean}}} \|z - z_\tau\|_{\mathcal{H}} \right] \leq \rho (\gamma \varepsilon + \delta),
\end{equation}

where \( \mathcal{H} \) is the RKHS associated with the feature encoder.
\label{theorem:upper bound on feature-manifold}
\end{theorem}

\begin{proof}
By the linearity of expectation and the definition of \( P_{\mathcal{M}_{\mathrm{mixed}}} \), we have:
\begin{equation}    
\begin{aligned}
&\quad\mathbb{E}_{z \sim P_{\mathcal{M}_{\mathrm{mixed}}}} \left[ \inf_{z_\tau} \|z - z_\tau\|_{\mathcal{H}} \right]\\ 
&= (1-\rho) \underbrace{\mathbb{E}_{z \sim P_{\mathcal{M}_{\mathrm{clean}}}} \left[ \inf_{z_\tau} \|z - z_\tau\|_{\mathcal{H}} \right]}_{= 0} \\
&+ \rho \mathbb{E}_{z \sim P_{\mathcal{M}_{\mathrm{triggered}}}} \left[ \inf_{z_\tau} \|z - z_\tau\|_{\mathcal{H}} \right].
\end{aligned}
\label{eq:(29)}
\end{equation}

Since clean samples $z\sim P_{\mathcal{M}_\text{clean}}$ lie on the target manifold, their distance minimum distance to the target manifold is zero. Therefore:
\begin{equation}    
\begin{aligned}
&\quad\mathbb{E}_{z \sim P_{\mathcal{M}_{\mathrm{mixed}}}} \left[ \inf_{z_\tau} \|z - z_\tau\|_{\mathcal{H}} \right]\\ 
&= \rho \mathbb{E}_{z \sim P_{\mathcal{M}_{\mathrm{triggered}}}} \left[ \inf_{z_\tau} \|z - z_\tau\|_{\mathcal{H}} \right].
\end{aligned}
\end{equation}

By Lemma \ref{lemma:diffeomorphism}, for any \( z_t \in \mathcal{M}_{\mathrm{triggered}} \), there exists \( z_\tau \in \mathcal{M}_{\mathrm{clean}} \) such that:
\begin{equation}
    \|z_t - z_\tau\|_{\mathcal{H}} \leq \delta' = \gamma \varepsilon + \delta.
\end{equation}

Hence,
\begin{equation}
    \inf_{z_\tau \in \mathcal{M}_{\mathrm{clean}}} \|z_t - z_\tau\|_{\mathcal{H}} \leq \delta'.
\end{equation}

Taking the expectation over \( P_{\mathcal{M}_{\mathrm{triggered}}} \), we obtain:
\begin{equation}
    \mathbb{E}_{z \sim P_{\mathcal{M}_{\mathrm{triggered}}}} \left[ \inf_{z_\tau} \|z - z_\tau\|_{\mathcal{H}} \right] \leq \delta'.
\end{equation}

Substituting into Eq.(\ref{eq:(29)}) yields:
\begin{equation}
    \mathbb{E}_{z \sim P_{\mathcal{M}_{\mathrm{mixed}}}} \left[ \inf_{z_\tau} \|z - z_\tau\|_{\mathcal{H}} \right] \leq \rho (\gamma \varepsilon + \delta).
\end{equation}
\end{proof}

\begin{theorem}[Upper Bound on the Discrepancy Between Poisoned and Clean Condensation Datasets]
Let \( \mathcal{T}_{y_\tau} \) denote the clean target-class dataset and \( \mathcal{T}_{\mathrm{mixed}} = \mathcal{T}_{y_\tau} \cup \mathcal{T}_{\mathrm{triggered}} \), where \( \mathcal{T}_{\mathrm{triggered}} \) consists of source-class samples \( x \in \mathcal{T}_{y_s} \) perturbed by a trigger generator \( G_\phi \) and relabeled as the target class.

Let \( \mathcal{S}_{\mathrm{clean}} \) and \( \mathcal{S}_{\mathrm{poison}} \) denote the condensation datasets distilled from \( \mathcal{T}_{y_\tau} \) and \( \mathcal{T}_{\mathrm{mixed}} \), respectively, by minimizing:
\begin{equation}
    \mathcal{S}^* = \arg\min_{\mathcal{S}} \mathrm{MMD}(\mathcal{T}, \mathcal{S}) + \lambda \mathcal{R}(\mathcal{S}),
\end{equation}

where \( \mathcal{T} \in \{ \mathcal{T}_{y_\tau}, \mathcal{T}_{\mathrm{mixed}} \} \), \( \lambda > 0 \), and \( \mathcal{R} \) is a strongly convex regularizer.

Under Assumptions \ref{assumption:Lipschitz Continuity}, \ref{assumption:local compactness}, and \ref{assumption:Inter-Class Boundary Proximi}, the MMD between \( \mathcal{S}_{\mathrm{clean}} \) and \( \mathcal{S}_{\mathrm{poison}} \) satisfies:
\[
\mathrm{MMD}(\mathcal{S}_{\mathrm{clean}}, \mathcal{S}_{\mathrm{poison}}) \leq \frac{L_f^2 \rho (\gamma \varepsilon + \delta)}{\lambda \mu_R}
\]
where \( \gamma = L_f \alpha \), \( \delta = \sup_{z_s \in \mathcal{M}_{\mathrm{source}}} \inf_{z_\tau \in \mathcal{M}_{\mathrm{clean}}} \|z_s - z_\tau\|_{\mathcal{H}} \), \( \rho \) is the poisoning rate, and \( \varepsilon \) bounds the input perturbation.
\end{theorem}

\begin{proof} By Theorem~\ref{theorem:upper bound on feature-manifold}:
\begin{equation}
    \mathbb{E}_{z \sim P_{\mathcal{M}_{\mathrm{mixed}}}} \left[ \inf_{z_\tau \in \mathcal{M}_{\mathrm{clean}}} \|z - z_\tau\|_{\mathcal{H}} \right] \leq \rho(\gamma \varepsilon + \delta).
\end{equation}

This inequality constrains the average deviation of the mixed distribution from the clean target manifold by $\rho(\gamma\varepsilon+\delta)$.

In RKHS, MMD can be expressed via the norm of mean embeddings:
\begin{equation}
    \mathrm{MMD}(\mathcal{T}_{y_\tau}, \mathcal{T}_\text{mixed}) = \|\mu_\text{clean} - \mu_\text{mixed}\|_{\mathcal{H}}.
\end{equation}

where 
\[\mu_\text{clean}=\mathbb{E}_{x\sim P_{\mathcal{T}_{y_\tau}}}[f_{\theta_f}(x)]
\]
\[\mu_\text{mixed}=\mathbb{E}_{x\sim P_{\mathcal{T}_{y_\text{mixed}}}}[f_{\theta_f}(x)]
\]

Using the decomposition, the mean embedding of the mixed distribution can be written as::
\begin{equation}
    \mu_\text{mixed} = (1-\rho)\mu_\text{clean} + \rho \mu_\text{triggered}
\end{equation}

we get:
\begin{equation}
    \mu_\text{clean} - \mu_\text{mixed} = \rho (\mu_\text{clean} - \mu_\text{triggered})
\end{equation}

Hence:
\begin{equation}
\begin{aligned}
    \mathrm{MMD}(\mathcal{T}_{y_\tau}, \mathcal{T}_\text{mixed}) &= \rho \|\mu_\text{clean} - \mu_\text{triggered}\|_{\mathcal{H}}\\
    &\leq \rho (\gamma \varepsilon + \delta)
\end{aligned}
\end{equation}

Let the clean and poisoned synthetic datasets, $\mathcal{S}_\text{clean}$ and $\mathcal{S}_\text{poison}$, be obtained by solving the following optimization problems:
\begin{equation}
\begin{aligned}
\mathcal{S}_\text{clean} &= \arg\min_{\mathcal{S}} \text{MMD}(\mathcal{T}_{y_\tau}, \mathcal{S}) + \lambda \mathcal{R}(\mathcal{S}), \\
\mathcal{S}_\text{poison} &= \arg\min_{\mathcal{S}} \text{MMD}(\mathcal{T}_\text{mixed}, \mathcal{S}) + \lambda \mathcal{R}(\mathcal{S})
\end{aligned}
\end{equation}

According to the first-order optimality condition, the solutions $\mathcal{S}_\text{clean}$ and $\mathcal{S}_\text{poison}$ satisfy:

\begin{equation}
\begin{aligned}
    \nabla \text{MMD}_\mathcal{S}(\mathcal{T}_{y_\tau}, \mathcal{S}_\text{clean})+\lambda\nabla\mathcal{R(\mathcal{S}_\text{clean})}&=0\\ 
    \nabla \text{MMD}_\mathcal{S}(\mathcal{T}_{y_\text{mixed}}, \mathcal{S}_\text{poison})+\lambda\nabla\mathcal{R(\mathcal{S}_\text{poison})}&=0
\end{aligned}
\end{equation}

Subtracting the optimality conditions:
\begin{equation}
\begin{aligned}
    \lambda(\nabla \mathcal{R}(\mathcal{S}_\text{clean}) - \nabla \mathcal{R}(\mathcal{S}_\text{poison})) &= \nabla \text{MMD}_\mathcal{S}(\mathcal{T}_\text{mixed}, \mathcal{S}_\text{poison})\\
    &- \nabla \text{MMD}_\mathcal{S}(\mathcal{T}_{y_\tau}, \mathcal{S}_\text{clean})
\end{aligned}
\end{equation}

Since \( \mathcal{R} \) is \( \mu_\mathcal{R} \)-strongly convex, we obtain:
\begin{equation}
\begin{aligned}
    &\quad\langle \nabla \mathcal{R}(\mathcal{S}_\text{clean}) - \nabla \mathcal{R}(\mathcal{S}_\text{poison}), \mathcal{S}_\text{clean} - \mathcal{S}_\text{poison} \rangle \\
    &\geq \mu_\mathcal{R} \|\mathcal{S}_\text{clean} - \mathcal{S}_\text{poison}\|^2
\end{aligned}
\end{equation}

Then, we can obtain:
\begin{equation}
\begin{aligned}
    &\quad\|\mathcal{S}_\text{clean} - \mathcal{S}_\text{poison}\|  \\
    &\leq\frac{\|\nabla_{\mathcal{S}} \text{MMD}(\mathcal{T}_{y_\tau}, \mathcal{S}_\text{clean}) - \nabla_{\mathcal{S}} \text{MMD}(\mathcal{T}_\text{mixed}, \mathcal{S}_\text{poison})\|}{\lambda \mu_\mathcal{R}}\\
    &\leq\frac{L_f\mathrm{MMD}(\mathcal{T}_{y_\tau}, \mathcal{T}_\text{mixed})}{\lambda \mu_\mathcal{R}}\\
    &\leq\frac{L_f\rho(\gamma\varepsilon+\delta)}{\lambda \mu_\mathcal{R}}
\end{aligned}
\end{equation}

According to Assumption~\ref{assumption:Lipschitz Continuity}:
\begin{equation}
\begin{aligned}
    \text{MMD}(\mathcal{S}_\text{clean}, \mathcal{S}_\text{poison}) &\leq L_f \|\mathcal{S}_\text{clean} - \mathcal{S}_\text{poison}\| \\
    &\leq \frac{L_f^2 \rho (\gamma \varepsilon + \delta)}{\lambda \mu_R}.
\end{aligned}
\end{equation}
\end{proof}

\section{Additional Experiments}
In dataset condensation, simple architectures such as ConvNet or AlexNetBN are typically employed as condensation networks, rather than more complex models. This design choice is motivated by several factors. First, computational efficiency and stability: simpler networks are faster and less resource-intensive to train, which is essential given the iterative optimization cycles required in dataset condensation. In contrast, deeper architectures substantially increase computational cost and introduce greater instability during optimization. Second, optimization tractability: simple models possess smoother and more navigable loss landscapes, facilitating the extraction of effective gradients from synthetic data. Complex architectures, with highly non-convex objectives, complicate this process and hinder optimization. Third, fairness and generality: the distilled data is intended to generalize across a range of architectures. Relying on a highly specialized, deep network risks overfitting the synthetic data to its unique characteristics. Employing a lightweight, generic model encourages the generation of broadly transferable synthetic datasets.

To further substantiate the choice of AlexNetBN as the condensation network, we report additional experimental results in the appendix. While ConvNet is widely adopted in dataset condensation for its simplicity, AlexNetBN introduces greater depth and batch normalization, offering a complementary evaluation of the distilled data’s robustness and generalizability. These experiments assess whether the performance patterns observed with ConvNet persist under a moderately more complex architecture, thereby strengthening the evidence for the reliability of the distilled datasets.

\subsection{Effectiveness on Different Datasets and Settings}
Firstly, for completeness, we report the results of the Naive attack in Table~\ref{tab:eff_sneakdoor_naive}.

\begingroup
\setlength{\tabcolsep}{3pt}

\begin{table*}[htbp]
\centering
\setlength{\tabcolsep}{13pt}
\caption{Effectiveness on Different Datasets}
\begin{adjustbox}{width=\textwidth}
{
\begin{tabular}{cc*{2}{c>{\columncolor{gray!25}}c}}
\toprule
\multirow{2}{*}{\textbf{Dataset}} & \multirow{2}{*}{\textbf{Method}}  
  & \multicolumn{2}{c}{\textbf{\textsc{Sneakdoor}}}  
  & \multicolumn{2}{c}{\textbf{NAIVE}} \\
& & CTA & ASR & CTA & ASR \\
\midrule
\multirow{4}{*}{\textbf{CIFAR10}}
& DM  & 0.626$\pm$0.001 & 0.989$\pm$0.000 & 0.632$\pm$0.001 & 0.113$\pm$0.012 \\
& DC  & 0.537$\pm$0.000 & 0.996$\pm$0.000 & 0.552$\pm$0.001 & 0.102$\pm$0.007 \\
& IDM & 0.643$\pm$0.002 & 0.975$\pm$0.001 & 0.652$\pm$0.001 & 0.103$\pm$0.006 \\
& DAM & 0.591$\pm$0.001 & 0.979$\pm$0.001 & 0.582$\pm$0.001 & 0.086$\pm$0.003 \\
\midrule
\multirow{4}{*}{\textbf{STL10}}
& DM  & 0.598$\pm$0.001 & 0.973$\pm$0.000 & 0.621$\pm$0.001 & 0.103$\pm$0.006 \\
& DC  & 0.565$\pm$0.001 & 0.998$\pm$0.001 & 0.583$\pm$0.001 & 0.090$\pm$0.007 \\
& IDM & 0.658$\pm$0.001 & 0.979$\pm$0.001 & 0.667$\pm$0.001 & 0.102$\pm$0.007 \\
& DAM & 0.532$\pm$0.001 & 0.992$\pm$0.001 & 0.549$\pm$0.001 & 0.088$\pm$0.009 \\
\midrule
\multirow{4}{*}{\textbf{FMNIST}}
& DM  & 0.876$\pm$0.001 & 0.998$\pm$0.000 & 0.887$\pm$0.001 & 0.090$\pm$0.008 \\
& DC  & 0.851$\pm$0.001 & 0.998$\pm$0.000 & 0.857$\pm$0.001 & 0.086$\pm$0.002 \\
& IDM & 0.877$\pm$0.001 & 1.000$\pm$0.000 & 0.887$\pm$0.001 & 0.093$\pm$0.007 \\
& DAM & 0.877$\pm$0.000 & 0.996$\pm$0.000 & 0.881$\pm$0.001 & 0.098$\pm$0.005 \\
\midrule
\multirow{4}{*}{\textbf{SVHN}}
& DM  & 0.800$\pm$0.000 & 1.000$\pm$0.000 & 0.799$\pm$0.000 & 0.111$\pm$0.006 \\
& DC  & 0.687$\pm$0.000 & 1.000$\pm$0.000 & 0.699$\pm$0.001 & 0.115$\pm$0.011 \\
& IDM & 0.831$\pm$0.001 & 0.986$\pm$0.001 & 0.840$\pm$0.000 & 0.122$\pm$0.010 \\
& DAM & 0.782$\pm$0.001 & 1.000$\pm$0.000 & 0.770$\pm$0.000 & 0.112$\pm$0.006 \\
\midrule
\multirow{4}{*}{\shortstack{\textbf{TINY}\\\textbf{IMAGENET}}}
& DM  & 0.503$\pm$0.001 & 1.000$\pm$0.000 & 0.497$\pm$0.002 & 0.070$\pm$0.002 \\
& DC  & 0.432$\pm$0.002 & 1.000$\pm$0.000 & 0.421$\pm$0.002 & 0.019$\pm$0.001 \\
& IDM & 0.517$\pm$0.004 & 1.000$\pm$0.000 & 0.501$\pm$0.008 & 0.042$\pm$0.004 \\
& DAM & 0.482$\pm$0.003 & 1.000$\pm$0.000 & 0.462$\pm$0.003 & 0.042$\pm$0.002 \\
\bottomrule
\end{tabular}
}
\end{adjustbox}
\label{tab:eff_sneakdoor_naive}
\end{table*}

\endgroup
Table~\ref{tab:eff_datasets_sneakdoor_naive_doorping} and ~\ref{tab:eff_datasets_sneakdoor_simple_relax} reports the ASR and CTA of different dataset condensation methods using AlexNetBN as the condensation network across multiple datasets. The results reveal how distilled data behaves under both clean and backdoor settings when applied to AlexNetBN. This provides a comprehensive view of each attack’s robustness and generalization in adversarial contexts.


\begin{table*}[htbp]
\centering
\caption{Effectiveness on Different Datasets condensed with AlexNetBN}
\begin{adjustbox}{width=\textwidth}
{
\begin{tabular}{cc*{3}{c>{\columncolor{gray!25}}c}}
\toprule
\multirow{2}{*}{\textbf{Dataset}} & \multirow{2}{*}{\textbf{Method}}
  & \multicolumn{2}{c}{\textbf{\textsc{Sneakdoor}}}
  & \multicolumn{2}{c}{\textbf{NAIVE}}
  & \multicolumn{2}{c}{\textbf{DOORPING}} \\
& & CTA & ASR & CTA & ASR & CTA & ASR \\
\midrule
\multirow{4}{*}{\textbf{CIFAR10}}
& DM  & 0.595$\pm$0.001 & 0.947$\pm$0.004 & 0.608$\pm$0.002 & 0.093$\pm$0.011 & 0.505$\pm$0.001 & 1.000$\pm$0.000 \\
& DC  & 0.222$\pm$0.001 & 0.003$\pm$0.001 & 0.140$\pm$0.001 & 0.000$\pm$0.000 & 0.319$\pm$0.007 & 0.000$\pm$0.000 \\
& IDM & 0.700$\pm$0.002 & 0.946$\pm$0.003 & 0.739$\pm$0.002 & 0.104$\pm$0.009 & 0.639$\pm$0.003 & 1.000$\pm$0.000 \\
& DAM & 0.606$\pm$0.001 & 0.721$\pm$0.013 & 0.609$\pm$0.001 & 0.096$\pm$0.010 & 0.565$\pm$0.001 & 1.000$\pm$0.000 \\
\midrule
\multirow{4}{*}{\textbf{STL10}}
& DM  & 0.562$\pm$0.001 & 0.993$\pm$0.000 & 0.573$\pm$0.004 & 0.104$\pm$0.010 & 0.557$\pm$0.004 & 1.000$\pm$0.000 \\
& DC  & 0.155$\pm$0.006 & 0.003$\pm$0.002 & 0.178$\pm$0.001 & 0.000$\pm$0.000 & 0.278$\pm$0.003 & 1.000$\pm$0.000 \\
& IDM & 0.723$\pm$0.002 & 0.986$\pm$0.002 & 0.729$\pm$0.003 & 0.100$\pm$0.007 & 0.646$\pm$0.003 & 1.000$\pm$0.000 \\
& DAM & 0.584$\pm$0.001 & 0.962$\pm$0.003 & 0.603$\pm$0.004 & 0.101$\pm$0.010 & 0.565$\pm$0.000 & 1.000$\pm$0.000 \\
\midrule
\multirow{4}{*}{\textbf{FMNIST}}
& DM  & 0.822$\pm$0.000 & 1.000$\pm$0.000 & 0.844$\pm$0.001 & 0.090$\pm$0.010 & 0.636$\pm$0.005 & 1.000$\pm$0.000 \\
& DC  & 0.287$\pm$0.000 & 0.000$\pm$0.000 & 0.172$\pm$0.003 & 0.320$\pm$0.018 & 0.516$\pm$0.010 & 1.000$\pm$0.000 \\
& IDM & 0.844$\pm$0.001 & 0.978$\pm$0.002 & 0.858$\pm$0.001 & 0.113$\pm$0.003 & 0.736$\pm$0.001 & 1.000$\pm$0.000 \\
& DAM & 0.831$\pm$0.003 & 1.000$\pm$0.000 & 0.821$\pm$0.002 & 0.100$\pm$0.003 & 0.758$\pm$0.003 & 1.000$\pm$0.000 \\
\midrule
\multirow{4}{*}{\textbf{SVHN}}
& DM  & 0.622$\pm$0.020 & 1.000$\pm$0.000 & 0.697$\pm$0.007 & 0.124$\pm$0.006 & 0.774$\pm$0.001 & 1.000$\pm$0.000 \\
& DC  & 0.108$\pm$0.001 & 0.984$\pm$0.001 & 0.095$\pm$0.001 & 0.069$\pm$0.010 & 0.379$\pm$0.006 & 1.000$\pm$0.000 \\
& IDM & 0.880$\pm$0.001 & 0.966$\pm$0.001 & 0.886$\pm$0.001 & 0.116$\pm$0.010 & 0.781$\pm$0.002 & 1.000$\pm$0.000 \\
& DAM & 0.672$\pm$0.006 & 0.999$\pm$0.000 & 0.701$\pm$0.002 & 0.112$\pm$0.008 & 0.593$\pm$0.003 & 1.000$\pm$0.000 \\
\midrule
\multirow{4}{*}{\shortstack{\textbf{TINY}\\\textbf{IMAGENET}}}
& DM  & 0.463$\pm$0.002 & 0.920$\pm$0.013 & 0.457$\pm$0.003 & 0.011$\pm$0.002 & 0.485$\pm$0.002 & 1.000$\pm$0.000 \\
& DC  & 0.247$\pm$0.003 & 1.000$\pm$0.000 & 0.269$\pm$0.005 & 0.013$\pm$0.003 & 0.260$\pm$0.004 & 0.000$\pm$0.000 \\
& IDM & 0.260$\pm$0.005 & 0.860$\pm$0.013 & 0.284$\pm$0.007 & 0.000$\pm$0.000 & 0.293$\pm$0.006 & 1.000$\pm$0.000 \\
& DAM & 0.442$\pm$0.006 & 0.972$\pm$0.010 & 0.430$\pm$0.013 & 0.010$\pm$0.001 & 0.419$\pm$0.010 & 1.000$\pm$0.000 \\
\bottomrule
\end{tabular}
}
\end{adjustbox}
\label{tab:eff_datasets_sneakdoor_naive_doorping}
\end{table*}

\vspace{1em}

\begin{table*}[htbp]
\centering
\caption{Effectiveness on Different Datasets condensed with AlexNetBN}
\begin{adjustbox}{width=\textwidth}
{
\begin{tabular}{cc*{3}{c>{\columncolor{gray!25}}c}}
\toprule
\multirow{2}{*}{\textbf{Dataset}} & \multirow{2}{*}{\textbf{Method}}
  & \multicolumn{2}{c}{\textbf{\textsc{Sneakdoor}}}
  & \multicolumn{2}{c}{\textbf{SIMPLE}}
  & \multicolumn{2}{c}{\textbf{RELAX}} \\
& & CTA & ASR & CTA & ASR & CTA & ASR \\
\midrule
\multirow{4}{*}{\textbf{CIFAR10}}
& DM  & 0.595$\pm$0.001 & 0.947$\pm$0.004 & 0.581$\pm$0.001 & 0.183$\pm$0.013 & 0.603$\pm$0.001 & 0.704$\pm$0.022 \\
& DC  & 0.222$\pm$0.001 & 0.003$\pm$0.001 & 0.169$\pm$0.002 & 0.000$\pm$0.000 & 0.152$\pm$0.001 & 0.047$\pm$0.018 \\
& IDM & 0.700$\pm$0.002 & 0.946$\pm$0.003 & 0.727$\pm$0.001 & 0.146$\pm$0.009 & 0.252$\pm$0.002 & 0.636$\pm$0.024 \\
& DAM & 0.606$\pm$0.001 & 0.721$\pm$0.013 & 0.584$\pm$0.001 & 0.204$\pm$0.024 & 0.591$\pm$0.002 & 0.978$\pm$0.004 \\
\midrule
\multirow{4}{*}{\textbf{STL10}}
& DM  & 0.562$\pm$0.001 & 0.993$\pm$0.000 & 0.544$\pm$0.002 & 0.092$\pm$0.007 & 0.550$\pm$0.003 & 0.706$\pm$0.010 \\
& DC  & 0.155$\pm$0.006 & 0.003$\pm$0.002 & 0.121$\pm$0.008 & 0.117$\pm$0.013 & 0.144$\pm$0.003 & 0.574$\pm$0.036 \\
& IDM & 0.723$\pm$0.002 & 0.986$\pm$0.002 & 0.724$\pm$0.003 & 0.102$\pm$0.013 & 0.719$\pm$0.002 & 0.668$\pm$0.029 \\
& DAM & 0.584$\pm$0.001 & 0.962$\pm$0.003 & 0.568$\pm$0.003 & 0.098$\pm$0.010 & 0.566$\pm$0.005 & 0.872$\pm$0.022 \\
\midrule
\multirow{4}{*}{\textbf{FMNIST}}
& DM  & 0.822$\pm$0.000 & 1.000$\pm$0.000 & 0.812$\pm$0.006 & 0.952$\pm$0.009 & 0.816$\pm$0.003 & 1.000$\pm$0.000 \\
& DC  & 0.287$\pm$0.000 & 0.000$\pm$0.000 & 0.161$\pm$0.001 & 0.895$\pm$0.018 & 0.171$\pm$0.001 & 0.646$\pm$0.033 \\
& IDM & 0.844$\pm$0.001 & 0.978$\pm$0.002 & 0.849$\pm$0.001 & 0.231$\pm$0.028 & 0.856$\pm$0.001 & 0.719$\pm$0.015 \\
& DAM & 0.831$\pm$0.003 & 1.000$\pm$0.000 & 0.806$\pm$0.002 & 0.482$\pm$0.128 & 0.811$\pm$0.002 & 1.000$\pm$0.000 \\
\midrule
\multirow{4}{*}{\textbf{SVHN}}
& DM  & 0.622$\pm$0.020 & 1.000$\pm$0.000 & 0.484$\pm$0.010 & 0.071$\pm$0.005 & 0.672$\pm$0.009 & 0.978$\pm$0.007 \\
& DC  & 0.108$\pm$0.001 & 0.984$\pm$0.001 & 0.157$\pm$0.006 & 0.060$\pm$0.006 & 0.137$\pm$0.004 & 0.119$\pm$0.027 \\
& IDM & 0.880$\pm$0.001 & 0.966$\pm$0.001 & 0.880$\pm$0.001 & 0.118$\pm$0.008 & 0.874$\pm$0.001 & 1.000$\pm$0.001 \\
& DAM & 0.672$\pm$0.006 & 0.999$\pm$0.000 & 0.693$\pm$0.006 & 0.092$\pm$0.007 & 0.692$\pm$0.003 & 0.996$\pm$0.003 \\
\midrule
\multirow{4}{*}{\shortstack{\textbf{TINY}\\\textbf{IMAGENET}}}
& DM  & 0.463$\pm$0.002 & 0.920$\pm$0.013 & 0.457$\pm$0.003 & 0.011$\pm$0.002 & 0.449$\pm$0.003 & 0.835$\pm$0.017 \\
& DC  & 0.247$\pm$0.003 & 1.000$\pm$0.000 & 0.200$\pm$0.008 & 0.000$\pm$0.000 & 0.259$\pm$0.002 & 0.471$\pm$0.023 \\
& IDM & 0.260$\pm$0.005 & 0.860$\pm$0.013 & 0.337$\pm$0.006 & 0.053$\pm$0.008 & 0.313$\pm$0.007 & 0.759$\pm$0.058 \\
& DAM & 0.442$\pm$0.006 & 0.972$\pm$0.010 & 0.443$\pm$0.007 & 0.013$\pm$0.002 & 0.441$\pm$0.004 & 0.787$\pm$0.027 \\
\bottomrule
\end{tabular}
}
\end{adjustbox}
\label{tab:eff_datasets_sneakdoor_simple_relax}
\end{table*}


Moreover, we have expanded our evaluation in two key directions: (1) \textbf{\textit{incorporating a larger, higher-resolution dataset}}, ImageNette (resolution $3\times224\times224$), as shown in Table~\ref{tab:imagenette}, and (2) \textbf{\textit{evaluating key parameters}} on STL10 (resolution $3\times96\times96$), including \textbf{\textit{ipc}} (the number of synthetic samples per clas), \textbf{\textit{perturbation bound $\varepsilon$}}, and \textbf{\textit{poisoning ratio}}, as shown in Table \ref{tab:ipc}, \ref{tab:bound}, and \ref{tab:ratio}.

Table~\ref{tab:imagenette} reports \textsc{Sneakdoor}’s attack performance under DM and DAM on the ImageNette dataset, demonstrating that \textbf{\textit{\textsc{Sneakdoor} remains effective on higher-resolution, larger-scale data}}. Due to computational resources constraints,  we could not include results for DC and IDM, as a single run with DC or IDM takes about three to four days, making full tuning impractical. We plan to include these results in a future version to provide a more complete picture of performance across algorithms and settings.

\begin{table}[htbp]
  \centering
  \small
  \caption{Attack Performance of \textsc{Sneakdoor} on the ImageNette Dataset.}
  \begin{tabular}{lccccc}
    \hline
    \\[-2ex]
    Method & ASR & CTA & PNSR & SSIM & IS \\
    \\[-2ex] \hline \\[-2ex]
    DM  & 0.9809±0.0000 & 0.5625±0.0007 & 68.62 & 0.6673 & 2.25e-4 \\
    DAM & 0.9429±0.0008 & 0.4598±0.0003 & 72.16 & 0.6814 & 2.08e-4 \\
    \\[-2ex] \hline
  \end{tabular}
  \label{tab:imagenette}
\end{table}

\begin{table}[htbp]
  \centering
  \small
  \caption{Impact of IPC on Attack Performance}
  \begin{tabular}{lcccccc}
    \hline\\[-2ex]
    Method & ipc & ASR & CTA & PSNR & SSIM & IS\\
    \\[-2ex] \hline \\[-2ex]
    DM   & 10 & 0.8735±0.0009 & 0.4347±0.0003 & 73.0381 & 0.8211 & 9.05e-5 \\
    DM   & 20 & 0.9872±0.0005 & 0.4882±0.0008 & 73.5021 & 0.7950 & 1.32e-4 \\
    DM   & 50 & 0.9725±0.0000 & 0.5979±0.0006 & 70.1216 & 0.8066 & 1.41e-4 \\
    IDM  & 10 & 0.9778±0.0015 & 0.5965±0.0004 & 74.1393 & 0.8199 & 1.05e-4 \\
    IDM  & 20 & 0.9573±0.0009 & 0.6217±0.0006 & 73.9608 & 0.8049 & 2.39e-4\\
    IDM  & 50 & 0.9790±0.0009 & 0.6582±0.0005 & 70.1548 & 0.7554 & 1.40e-4\\
    DAM  & 10 & 0.8910±0.0015 & 0.3678±0.0006 & 73.6366 & 0.8106 & 9.21e-5 \\
    DAM  & 20 & 0.8902±0.0025 & 0.4522±0.0004 & 73.8535 & 0.8146 & 9.22e-5 \\
    DAM  & 50 & 0.9918±0.0006 & 0.5324±0.0007 & 73.7877 & 0.8245 & 9.14e-5 \\
    DC   & 10 & 0.9258±0.0035 & 0.4675±0.0006 & 73.1598 & 0.8072 & 9.54e-5 \\
    DC   & 20 & 0.9243±0.0035 & 0.5282±0.0002 & 73.0987 & 0.8018 & 9.05e-5 \\
    DC   & 50 & 0.9975±0.0008 & 0.5653±0.0011 & 71.2365 & 0.7550 & 7.26e-5 \\
    \\[-2ex] \hline
  \end{tabular}
  \label{tab:ipc}
\end{table}

As shown in Table~\ref{tab:ipc}, varying ipc notably affects CTA, while ASR and STE metrics (PSNR, SSIM, IS) remain relatively stable. This is expected, as fewer samples per class reduce the fidelity of clean distribution modeling, impacting generalization. In contrast, ASR stays high across ipc values, indicating that once embedded, the backdoor remains effective even with limited data. STE metrics also show minimal change, suggesting the perturbations remain visually subtle and robust.

As shown in Table~\ref{tab:bound}, increasing the perturbation bound $\varepsilon$ improves ASR but reduces STE, as reflected in lower PSNR, SSIM, and IS. This is expected, since a larger $\varepsilon$ allows stronger and more noticeable triggers, enhancing attack success at the expense of stealth. Notably, CTA remains stable across $\varepsilon$ values, indicating that stronger triggers do not significantly harm generalization on clean data. These results highlight a trade-off between ASR and STE controlled by $\varepsilon$.

\begin{table}[htbp]
  \centering
  \small
  \caption{Impact of Perturbation Bound $\varepsilon$ on Attack Performance}
  \begin{tabular}{lcccccc}
    \hline\\[-2ex]
    Method & $\varepsilon$ & ASR & CTA & PSNR & SSIM & IS\\
    \\[-2ex] \hline \\[-2ex]
    DM   & 0.1 & 0.7755±0.0049 & 0.6045±0.0009 & 82.1241 & 0.9548 & 2.97e-5 \\
    DM   & 0.2 & 0.9332±0.0006 & 0.5824±0.0008 & 76.9565 & 0.8769 & 5.46e-5 \\
    DM   & 0.3 & 0.9732±0.000 & 0.5981±0.0010 & 74.0076 & 0.7963 & 6.32e-5 \\
    IDM  & 0.1 & 0.5400±0.0076 & 0.6627±0.0010 & 78.7475 & 0.7914 & 1.14e-4 \\
    IDM  & 0.2 & 0.7905±0.0073 & 0.6624±0.0013 & 76.4274 & 0.7931 & 1.30e-4 \\
    IDM  & 0.3 & 0.9790±0.0009 & 0.6582±0.0005 & 70.1548 & 0.8054 & 1.40e-4 \\
    DAM  & 0.1 & 0.6785±0.0022 & 0.5278±0.0012 & 82.0221 & 0.9594 & 3.06e-5 \\
    DAM  & 0.2 & 0.8715±0.0015 & 0.5389±0.0007 & 76.8882 & 0.8916 & 5.51e-5 \\
    DAM  & 0.3 & 0.9918±0.0006 & 0.5324±0.0007 & 73.7877 & 0.8245 & 9.14e-5\\
    DC   & 0.1 & 0.6128±0.004 & 0.5743±0.0002 & 78.8841 & 0.7633 & 7.54e-5 \\
    DC   & 0.2 & 0.7828±0.0056 & 0.58±0.0011 & 73.3082 & 0.5337 & 1.06e-4 \\
    DC   & 0.3 & 0.9980 ± 0.0010 & 0.5650±0.0010 & 71.2365 & 0.5551 & 7.25e-5 \\
    \\[-2ex] \hline
  \end{tabular}
  \label{tab:bound}
\end{table}

\begin{table}[htbp]
  \centering
  \small
  \caption{Impact of Poisoning Ratio on Attack Performance}
  \begin{tabular}{lcccccc}
    \hline\\[-2ex]
    Method & poison ratio & ASR & CTA & PSNR & SSIM & IS \\
    \\[-2ex]\hline\\[-2ex]
    DM   & 0.10 & 0.8810±0.0020 & 0.5986±0.001 & 74.0086 & 0.8285 & 8.82e-5 \\
    DM   & 0.25 & 0.8970±0.0019 & 0.6009±0.0009 & 73.7735 & 0.7942 & 9.55e-5 \\
    DM   & 0.5 & 0.9725±0.0000 & 0.5979±0.0006 & 73.0076 & 0.7963 & 1.14e-4 \\
    IDM  & 0.10 & 0.8205±0.0026 & 0.6645±0.0015 & 74.0362 & 0.7803 & 2.61e-4 \\
    IDM  & 0.25 & 0.8615±0.0044 & 0.6592±0.0007 & 70.2375 & 0.7788 & 1.33e-4 \\
    IDM  & 0.5 & 0.9790±0.0009 & 0.6582±0.0005 & 70.1548 & 0.7554 & 1.40e-4 \\
    DAM  & 0.10 & 0.5073±0.0035 & 0.5526±0.0003 & 74.2949 & 0.8200 & 8.10e-5 \\
    DAM  & 0.25 & 0.7820±0.0017 & 0.5488±0.0006 & 73.5737 & 0.8429 & 1.11e-4 \\
    DAM  & 0.5 & 0.9918±0.0006 & 0.5324±0.0007 & 73.7877 & 0.8245 & 9.14e-5 \\
    DC   & 0.10 & 0.7912±0.0041 & 0.5745±0.0007 & 69.7258 & 0.5573 & 1.32e-4 \\
    DC   & 0.25 & 0.8627±0.0031 & 0.5851±0.0005 & 70.4030 & 0.5113 & 1.49e-4 \\
    DC   & 0.5 & 0.9980±0.0010 & 0.5650±0.0010 & 71.2365 & 0.5551 & 7.25e-5 \\
    \\[-2ex]\hline
  \end{tabular}
  \label{tab:ratio}
\end{table}

As shown in Table~\ref{tab:ratio}, increasing the poisoning ratio improves the ASR, which aligns with the intuition that more poisoned samples enhance the trigger’s influence in the condensed dataset. However, this improvement comes with a slight degradation in CTA. Interestingly, the decline in CTA is relatively limited even at higher poisoning ratios (\textit{e.g.}, 0.5), suggesting that the trigger's interference with the clean distribution remains modest. Nevertheless, the reliance on a relatively high poisoning ratio to achieve optimal attack effectiveness highlights a limitation of the current approach.

\subsection{Stealthiness on CIFAR10, SVHN, and FMNIST}

We have included stealthiness for the remaining datasets, \textit{i.e.}, CIFAR10, SVHN, and FMNIST. These additional results offer a comprehensive assessment of \textsc{Sneakdoor}'s visual imperceptibility across diverse datasets. Notably, we omit the Inception Score (IS) evaluation for FMNIST because it is a single-channel (grayscale) dataset, which is incompatible with the standard IS computation that relies on a pre-trained Inception network trained on RGB images. Applying IS directly to grayscale data would yield unreliable and uninformative results. 

\begin{table}[htbp]
  \centering
  \small
  \caption{PSNR, SSIM, and IS on CIFAR10, SVHN, and FMNIST}
  \label{tab:stealthiness_all}
  \begin{adjustbox}{width=\textwidth}
    {
  \begin{tabular}{ll|ccc|ccc|ccc}
    \toprule
    Method   & Backdoor  & \multicolumn{3}{c|}{CIFAR-10}      & \multicolumn{3}{c|}{SVHN}          & \multicolumn{3}{c}{FMNIST} \\
             &           & PSNR     & SSIM     & IS      & PSNR     & SSIM     & IS      & PSNR     & SSIM     & IS      \\
    \midrule
    \multirow{5}{*}{DM}
      & \textsc{Sneakdoor} & 73.94 & 0.61 & 5.80e-05 & 74.68 & 0.77 & 3.90e-05 & 58.41 & 0.39 & -- \\
      & Doorping  & 59.85 & 0.08 & 2.30e-04 & 60.27 & 0.08 & 2.08e-04 & 55.68 & 0.12 & -- \\
      & Relax     & 60.97 & -0.01 & 2.48e-04 & 61.47 & -0.14 & 2.45e-04 & 51.88 & -0.07 & -- \\
      & naive     & 63.67 & 0.15 & 3.56e-04 & 62.27 & 0.10 & 4.60e-04 & 54.15 & 0.10 & -- \\
      & Simple    & 60.98 & 0.69 & 8.10e-05 & 61.59 & 0.74 & 7.95e-05 & 54.01 & 0.00 & -- \\
    \midrule
    \multirow{5}{*}{DC}
      & \textsc{Sneakdoor} & 70.48 & 0.46 & 7.10e-05 & 73.15 & 0.42 & 8.10e-05 & 57.39 & 0.24 & -- \\
      & Doorping  & 59.22 & 0.05 & 2.43e-04 & 61.25 & 0.06 & 2.00e-04 & 60.11 & 0.52 & -- \\
      & Relax     & 61.37 & 0.04 & 2.38e-04 & 62.17 & -0.04 & 2.43e-04 & 52.15 & -0.11 & -- \\
      & naive     & 64.46 & 0.18 & 3.62e-04 & 60.45 & 0.04 & 4.92e-04 & 54.21 & 0.06 & -- \\
      & Simple    & 60.74 & 0.66 & 8.70e-05 & 61.44 & 0.72 & 8.08e-05 & 53.99 & 0.00 & -- \\
    \midrule
    \multirow{5}{*}{IDM}
      & \textsc{Sneakdoor} & 74.88 & 0.77 & 4.40e-05 & 72.19 & 0.68 & 6.30e-05 & 57.16 & 0.10 & -- \\
      & Doorping  & 59.23 & 0.10 & 2.23e-04 & 59.66 & 0.06 & 2.17e-04 & 57.26 & 0.06 & -- \\
      & Relax     & 61.18 & 0.02 & 2.46e-04 & 61.17 & -0.20 & 2.70e-04 & 52.04 & -0.08 & -- \\
      & naive     & 64.23 & 0.14 & 3.44e-04 & 62.05 & 0.07 & 5.02e-04 & 54.15 & 0.05 & -- \\
      & Simple    & 61.05 & 0.69 & 8.60e-05 & 61.21 & 0.70 & 8.00e-05 & 54.23 & 0.00 & -- \\
    \midrule
    \multirow{5}{*}{DAM}
      & \textsc{Sneakdoor} & 74.40 & 0.74 & 4.50e-05 & 78.91 & 0.74 & 4.30e-05 & 57.39 & 0.24 & -- \\
      & Doorping  & 59.52 & 0.08 & 1.62e-04 & 59.67 & 0.08 & 1.05e-04 & 57.16 & 0.10 & -- \\
      & Relax     & 61.19 & 0.02 & 2.31e-04 & 62.36 & -0.24 & 2.04e-04 & 51.83 & -0.10 & -- \\
      & naive     & 62.99 & 0.13 & 4.53e-04 & 60.43 & 0.04 & 5.39e-04 & 55.07 & 0.12 & -- \\
      & Simple    & 60.85 & 0.64 & 8.70e-05 & 61.78 & 0.75 & 7.95e-05 & 54.07 & 0.00 & -- \\
    \bottomrule
  \end{tabular}
  }
  \end{adjustbox}
\end{table}

\subsection{Effectiveness on Cross Architectures}

We further include cross-architecture evaluations with AlexNetBN. This setting tests the transferability of the backdoor attack to a moderately different network from the condensation model. The results offer additional evidence of the generalization and robustness of \textsc{Sneakdoor} across architectures. This property is critical for practical deployment in real-world scenarios.


\begingroup
\setlength{\tabcolsep}{3pt}

\begin{table*}[ht]
\centering
\caption{Cross‐architecture CTA and ASR condensed with AlexNetBN}
\begin{adjustbox}{width=\textwidth}
{
\begin{tabular}{cc*{4}{c>{\columncolor{gray!25}}c}}
\toprule
\multirow{2}{*}{\textbf{Dataset}} & \multirow{2}{*}{\textbf{Network}}
  & \multicolumn{2}{c}{\textbf{DM}}
  & \multicolumn{2}{c}{\textbf{DC}}
  & \multicolumn{2}{c}{\textbf{IDM}}
  & \multicolumn{2}{c}{\textbf{DAM}} \\
 & & CTA & ASR & CTA & ASR & CTA & ASR & CTA & ASR \\
\midrule
\multirow{3}{*}{\textbf{CIFAR10}}
 & VGG11   & 0.544$\pm$0.000 & 0.961$\pm$0.000 
           & 0.209$\pm$0.000 & 0.009$\pm$0.000 
           & 0.673$\pm$0.000 & 0.945$\pm$0.001 
           & 0.542$\pm$0.000 & 0.733$\pm$0.001 \\
 & ResNet  & 0.495$\pm$0.001 & 0.915$\pm$0.002 
           & 0.186$\pm$0.000 & 0.009$\pm$0.000 
           & 0.671$\pm$0.001 & 0.926$\pm$0.001 
           & 0.500$\pm$0.001 & 0.491$\pm$0.001 \\
 & ConvNet & 0.585$\pm$0.001 & 0.807$\pm$0.002 
           & 0.216$\pm$0.001 & 0.004$\pm$0.001 
           & 0.638$\pm$0.001 & 0.951$\pm$0.002 
           & 0.582$\pm$0.001 & 0.457$\pm$0.005 \\
\midrule
\multirow{3}{*}{\textbf{STL10}}
 & VGG11   & 0.527$\pm$0.001 & 0.921$\pm$0.000 
           & 0.195$\pm$0.001 & 0.012$\pm$0.001 
           & 0.694$\pm$0.000 & 0.947$\pm$0.002 
           & 0.547$\pm$0.001 & 0.924$\pm$0.002 \\
 & ResNet  & 0.413$\pm$0.001 & 0.999$\pm$0.000 
           & 0.160$\pm$0.001 & 0.011$\pm$0.001 
           & 0.644$\pm$0.001 & 0.991$\pm$0.001 
           & 0.445$\pm$0.002 & 0.995$\pm$0.000 \\
 & ConvNet & 0.532$\pm$0.000 & 0.841$\pm$0.002 
           & 0.180$\pm$0.000 & 0.152$\pm$0.005 
           & 0.693$\pm$0.001 & 0.828$\pm$0.011 
           & 0.555$\pm$0.001 & 0.997$\pm$0.001 \\
\midrule
\multirow{3}{*}{\shortstack{\textbf{TINY}\\\textbf{IMAGENET}}}
 & VGG11   & 0.427$\pm$0.001 & 0.920$\pm$0.000 
           & 0.174$\pm$0.002 & 0.860$\pm$0.000 
           & 0.435$\pm$0.003 & 0.588$\pm$0.024 
           & 0.437$\pm$0.002 & 0.960$\pm$0.000 \\
 & ResNet  & 0.361$\pm$0.002 & 0.800$\pm$0.000 
           & 0.227$\pm$0.002 & 0.716$\pm$0.008 
           & 0.228$\pm$0.004 & 0.360$\pm$0.036 
           & 0.391$\pm$0.002 & 1.000$\pm$0.000 \\
 & ConvNet & 0.443$\pm$0.003 & 0.604$\pm$0.008 
           & 0.217$\pm$0.003 & 0.932$\pm$0.010 
           & 0.335$\pm$0.009 & 0.604$\pm$0.015 
           & 0.430$\pm$0.004 & 0.884$\pm$0.015 \\
\bottomrule
\end{tabular}
}
\end{adjustbox}
\label{tab:alexnetbn_cross_architecture}
\end{table*}

\endgroup

\subsection{Visual Analysis of Trigger Stealthiness}
We provide visualizations of original images after injecting the trigger during inference. Figure \ref{fig:stl10_bd} illustrates the effect following trigger injection. The images demonstrate the trigger’s subtlety and stealthiness. Changes to the original images are minimal and barely perceptible. Despite this, the trigger effectively activates the backdoor in the model. These visual results emphasize the challenge of detecting such backdoors through simple inspection. They also underscore the importance of robust defenses against stealthy triggers.

\begin{figure}
    \centering
    \includegraphics[width=0.3\linewidth]{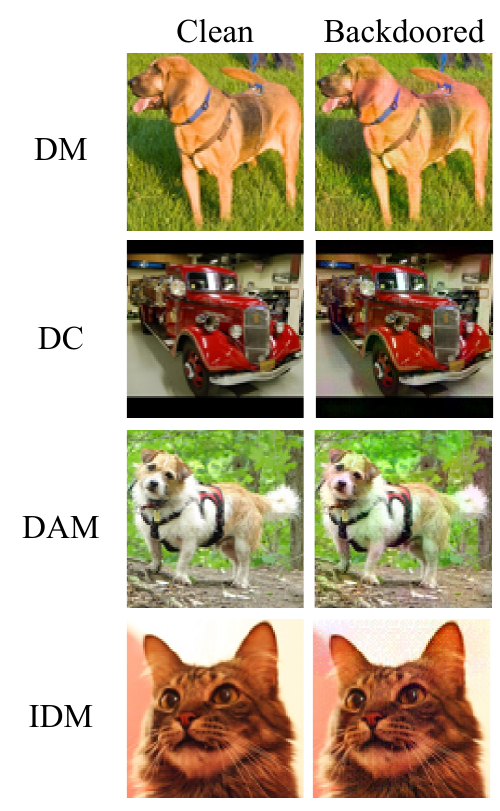}
    \caption{STL10 Stealthiness Illustration}
    \label{fig:stl10_bd}
    \vspace{-5pt}
\end{figure}

\subsection{Hyper-parameter Settings}
We have provided the full set of optimization hyperparameters used for \textsc{Sneakdoor} on the STL10 dataset across four condensation baselines: DM, DC, IDM, and DAM, including learning rates, number of epochs, batch sizes, etc. These details are listed in Tab.5 - Tab.8, allowing replication of our experiments. In addition, we will release the full source code in a future version of the paper. This will include the complete training pipeline for both the trigger generator and dataset condensation procedures. Our goal is to ensure that the community can easily reproduce and extend our work. 

The overall method is divided into four stages:

1. Training the Surrogate Model. The surrogate model serves two key purposes: (i) estimating inter-class boundary vulnerability (ICBV), and (ii) guiding the training of the trigger generator.

\begin{table}[htbp]
\centering
\caption{Hyperparameters for Surrogate Model Training}
\begin{tabular}{ll}
\toprule
\textbf{Hyperparameter} & \textbf{Value} \\
\midrule
Optimizer & SGD \\
Batch size & 256 \\
Learning rate & 0.01 \\
Momentum & 0.9 \\
Weight decay & 0.0005 \\
Epochs & 50 \\
\bottomrule
\end{tabular}
\end{table}

2. Training the Trigger Generator $G_\phi$. The generator learns to produce input-aware perturbations that cause misclassification.

\begin{table}[htbp]
\centering
\caption{Hyperparameters for Trigger Generator Training}
\begin{tabular}{ll}
\toprule
\textbf{Hyperparameter} & \textbf{Value} \\
\midrule
Learning rate & 5e-5 \\
Perturbation scaling factor $\alpha$ & 0.25 \\
Maximum perturbation bound $\varepsilon$ & 0.5 \\
\bottomrule
\end{tabular}
\end{table}

3. Malicious Condensation. This phase incorporates the trigger signal into the synthetic dataset via a standard condensation framework.

\begin{table}[htbp]
\centering
\caption{Hyperparameters for Malicious Dataset Condensation}
\begin{tabular}{ll}
\toprule
\textbf{Hyperparameter} & \textbf{Value} \\
\midrule
Images per class (IPC) & 50 \\
Condensation epochs & 20000 \\
Synthesis learning rate & 1.0 \\
Batch size & 256 \\
Optimizer & Adam \\
\bottomrule
\end{tabular}
\end{table}

4. Downstream Model Training. Standard training on the poisoned condensed dataset using typical optimization settings.

\begin{table}[htbp]
\centering
\caption{Hyperparameters for Downstream Model Training}
\begin{tabular}{ll}
\toprule
\textbf{Hyperparameter} & \textbf{Value} \\
\midrule
Optimizer & SGD \\
Batch size & 256 \\
Learning rate & 0.01 \\
Momentum & 0.9 \\
Weight decay & 0.0005 \\
Epochs & 10000 \\
\bottomrule
\end{tabular}
\end{table}

\end{document}